\documentclass{svmult}
\usepackage{mathptmx}
\usepackage{helvet}
\usepackage{courier}
\usepackage{type1cm}
\usepackage{graphicx}
\usepackage{epstopdf}
\usepackage{multicol}
\usepackage[bottom]{footmisc}
\usepackage{url}
% ...

\usepackage{amsmath,amssymb,epsfig}
\usepackage{colortbl}

% ...
\newtheorem{algorithm}{Algorithm}

\newcommand{\ceil}[1]{\left\lceil #1 \right\rceil}

\newcommand{\be}{\begin{eqnarray}}
\newcommand{\ee}{\end{eqnarray}}

\newcommand{\ba}{\begin{array}}
\newcommand{\ea}{\end{array}}
\newcommand{\bs}{\begin{align}\begin{split}\nonumber}
\newcommand{\bsnumber}{\begin{align}\begin{split}}
\newcommand{\es}{\end{split}\end{align}}

\renewcommand{\(}{\left(}
\renewcommand{\)}{\right)}
\renewcommand{\[}{\left[}
\renewcommand{\]}{\right]}
\renewcommand{\hat}{\widehat}

\newcommand{\cc}{\mathbf{c}}

\newcommand{\En}{\mathbb{E}_n}
\newcommand{\Ep}{E}

\newcommand{\sign}{ {\rm sign}}

\def\mmu#1{{\mu(#1)}}
\def\RR{{\Bbb{R}}}

\def\supp{{\rm support}}

\def\cc{{\bar c}}

\begin{document}

\title*{High Dimensional Sparse Econometric Models: An Introduction}
 \titlerunning{HDSM in Econometrics}
\author{Alexandre Belloni and Victor Chernozhukov}
\authorrunning{Belloni and Chernozhukov}
\institute{Alexandre Belloni \at Duke University, Fuqua School of Business, 100 Fuqua Drive, Durham, NC, \email{abn5@duke.edu}
\and Victor Chernozhukov \at Massachusetts Institute of Technology, Department of Economics, 50 Memorial Drive, Cambridge, MA \email{vchern@mit.edu}}

\maketitle

\abstract{In this chapter we discuss conceptually high dimensional sparse econometric models as well as estimation of these models using
$\ell_1$-penalization and post-$\ell_1$-penalization methods. Focusing on linear and nonparametric regression frameworks, we discuss various econometric examples,
present basic theoretical results, and illustrate the concepts and methods with Monte Carlo simulations and an empirical application. In the application, we examine and confirm the empirical validity of the Solow-Swan model for international economic growth.}

\section{The High Dimensional Sparse Econometric Model}

We consider linear, high dimensional sparse (HDS) regression models in econometrics.  The HDS regression model has a large number of regressors $p$, possibly much larger than the sample size $n$, but only a relatively small number $s < n$ of  these regressors are important for capturing accurately the main features of the regression function.  The latter assumption makes it possible to estimate these models effectively by searching for approximately the right set of the regressors, using $\ell_1$-based penalization methods. In this chapter we will review the basic
theoretical properties of these procedures, established in the works of \cite{BickelRitovTsybakov2009,CandesTao2007,MY2007,Lounici2008,BC-PostLASSO,Koltchinskii2009,vdGeer,ZhaoYu2006,ZhangHuang2006},  among others (see \cite{RigolletTsybakov2010,BC-PostLASSO} for a detailed literature review). In this section, we review the modeling foundations as well as motivating examples for these procedures, with emphasis on applications in econometrics.

Let us first consider an exact or parametric HDS regression model, namely,
\begin{equation}\label{Def:P}
y_i = x_i'\beta_0 + \varepsilon_i, \ \ \epsilon_i \sim N(0, \sigma^2), \ \ \beta_0 \in \Bbb{R}^p, \  i=1,\ldots, n,\end{equation} where $y_i$'s are observations of the response variable, $x_i$'s are observations of $p$-dimensional fixed regressors, and $\epsilon_i$'s are i.i.d. normal disturbances, where  possibly $p \geqslant n$. The key assumption of the exact model is  that the true parameter value $\beta_0$ is sparse, having only $s<n$ non-zero components with support denoted by
 \begin{equation} \label{Eq:T}
T=\text{ support} (\beta_0) \subset \{1,\ldots,p\}.
\end{equation}
Next let us consider an approximate or nonparametric HDS model. To this end, let us introduce the regression model
\begin{equation}\label{Def:NP}
y_i = f(z_i) + \varepsilon_i, \ \   \epsilon_i \sim N(0, \sigma^2),  \ \  i=1,\ldots, n,
\end{equation}
where $y_i$ is the outcome, $z_i$ is a vector of elementary fixed regressors,  $z \mapsto f(z)$ is the true, possibly non-linear, regression function, and $\varepsilon_i$'s are i.i.d. normal disturbances.  We can convert this model into an approximate HDS model by writing
\begin{equation}\label{Def:NP}
y_i = x_i' \beta_0 + r_i + \varepsilon_i,  \ \ i =1,\ldots, n,
\end{equation}
where $x_i=P(z_i)$ is a $p$-dimensional regressor formed from the elementary regressors by applying, for example, polynomial or spline transformations, $\beta$ is a conformable parameter vector, whose ``true" value $\beta_0$ has only $s<n$ non-zero components with support denoted as in (\ref{Eq:T}), and $r_i:=r(z_i)= f(z_i) - x_i'\beta_0$ is the approximation error.  We shall define the true value $\beta_0$ more precisely in the next section. For now, it is important to note only that we assume there exists a value $\beta_0$ having only $s$ non-zero components that sets the approximation error $r_i$ to be small.

Before considering estimation, a natural question is whether exact or approximate HDS models make sense in econometric applications.
In order to answer this question it is helpful to consider the following example, in which we abstract from estimation completely and only
ask whether it is possible to accurately describe some structural econometric function $f(z)$ using a low-dimensional approximation
of the form $P(z)'\beta_0$. In particular, we are interested in improving upon the conventional low-dimensional approximations.

\textbf{Example 1: Sparse Models for Earning Regressions}.   In this example we consider a model for the conditional
expectation of log-wage $y_i$ given education $z_i$, measured in years of schooling.
Since measured education takes on a finite number of years,  we can expand  the conditional expectation of wage $y_i$ given education $z_i$:\begin{equation}\label{general} E[y_i|z_i] = \sum_{j=1}^p
\beta_{0j} P_j(z_i),
 \end{equation}
 using some dictionary of approximating functions $P_1(z_i),\ldots, P_p(z_i)$, such as polynomial or spline transformations in $z_i$ and/or indicator variables for levels of
 $z_i$. In fact, since we can consider an overcomplete dictionary, the representation of the function may not be unique, but this is not important for our purposes.

A conventional sparse approximation employed in econometrics is, for example,
\begin{equation}\label{conventional}
f(z_i):=E[y_i|z_i] = \tilde \beta_1P_1(z_i) + \cdots + \tilde \beta_s P_s(z_i)  + \tilde r_i,
 \end{equation}
where the $P_j$'s are low-order polynomials or splines, with typically $s=4$ or $5$ terms, but there is no guarantee that the approximation error $\tilde r_i$ in this case is small, or that these particular polynomials form the best possible $s$-dimensional approximation. Indeed, we might expect  the function
$E[y_i|z_i]$ to exhibit oscillatory behavior near the schooling levels associated with advanced degrees, such as MBA or MD.  Low-degree polynomials may not be able to capture this behavior very well, resulting in large approximation errors $\tilde r_i$'s.

Therefore, the question is: With the same number of parameters, can we  find a much
better approximation?  In other words, can we find   some higher-order terms in the expansion (\ref{general}) which will provide a higher-quality approximation? More specifically, can we construct an approximation
\begin{equation}\label{eq:sparse}
f(z_i):=E[y_i|z_i] =  \beta_{k_1}P_{k_1}(z_i) + \cdots + \beta_{k_s} P_{k_s}(z_i) + r_i,
 \end{equation}
for some regressor indices $k_1,\ldots, k_s$ selected from $\{1,\ldots,p\}$, that is accurate and much better than (\ref{conventional}), in the sense of having a much smaller approximation error $r_i$?

Obviously the answer to the latter question depends on how complex the behavior of the true regression function (\ref{general}) is.
If the behavior is not complex, then low-dimensional approximation should be accurate. Moreover, it is clear that
the second approximation (\ref{eq:sparse}) is weakly better than the first (\ref{conventional}), and can be much better if there are some important high-order terms
in (\ref{general}) that are completely missed by the first approximation.  Indeed, in the context of the earning function example, such important high-order terms could capture abrupt positive changes in earning associated with advanced degrees such as MBA or MD.  Thus, the answer to the question depends strongly on the empirical context.

Consider for example the earnings of prime age white males in the 2000 U.S. Census (see e.g., Angrist, Chernozhukov and Fernandez-Val \cite{ACF2006}). Treating this data as the population data, we can then compute $f(z_i)=E[y_i|z_i]$ without error. Figure \ref{Fig:Wage} plots this function. (Of course, such a strategy is not generally available in the empirical work, since the population data are generally not available.)  We then construct two sparse approximations and also plot them in Figure \ref{Fig:Wage}: the first is the conventional one, of the form (\ref{conventional}), with $P_1, \ldots, P_s$ representing an  $(s-1)$-degree polynomial, and the second is an approximation of the form (\ref{eq:sparse}), with $P_{k_1}$, \ldots, $P_{k_s}$ consisting of a constant, a linear term, and two linear splines terms with knots located at 16 and 19 years of schooling (in the case of $s=5$ a third knot is located at 17).  In fact, we find the latter approximation automatically using $\ell_1$-penalization methods, although in this special case we could construct such an approximation just by eye-balling Figure \ref{Fig:Wage} and noting that most of the function is described by a linear function, with a few abrupt changes that can be captured by linear spline terms that induce large changes in slope near 17 and 19 years of schooling.  Note that an exhaustive search for a low-dimensional approximation requires looking at a very large set of models. We avoided this exhaustive search by using $\ell_1$-penalized least squares (LASSO), which penalizes the size of the model through the sum of absolute values of regression coefficients. Table \ref{Table:WageSparse} quantifies the performance of the different sparse approximations. (Of course, a simple strategy of eye-balling also works in this simple illustrative setting, but clearly does not apply to more general examples
with several conditioning variables $z_i$, for example, when we want to condition on education, experience, and age.) \qed

\begin{center}
\begin{table}
\begin{center}
\begin{tabular}{rcccccc}
\hline \hline
 Sparse
Approximation & & $s$ & & {   $L_2$ error } & & {  $L_\infty$ error } \\ \hline
Conventional  & & 4 & & 0.1212 & &    0.2969          \\
Conventional  & & 5 & & 0.1210 &  &   0.2896          \\
LASSO  & & 4 & & 0.0865 &  &  0.1443             \\
LASSO  & & 5 & & 0.0752 &  & 0.1154             \\
Post-LASSO  & & 4 & & 0.0586 &  &  0.1334             \\
Post-LASSO  & & 5 & & 0.0397 &  &   0.0788             \\

 \hline
\hline
\\
\end{tabular}
\end{center}\caption{Errors of Conventional and the LASSO-based Sparse Approximations of the Earning Function. The LASSO estimator minimizes the least squares criterion plus the $\ell_1$-norm of the coefficients scaled by a penalty parameter $\lambda$. As shown later, it turns out to have only a few non-zero components. The Post-LASSO estimator minimizes the least squares criterion over the non-zero components selected by the LASSO estimator.}\label{Table:WageSparse}
\end{table}
\end{center}

\begin{figure}[!h]
\centering
\includegraphics[width=0.7\textwidth]{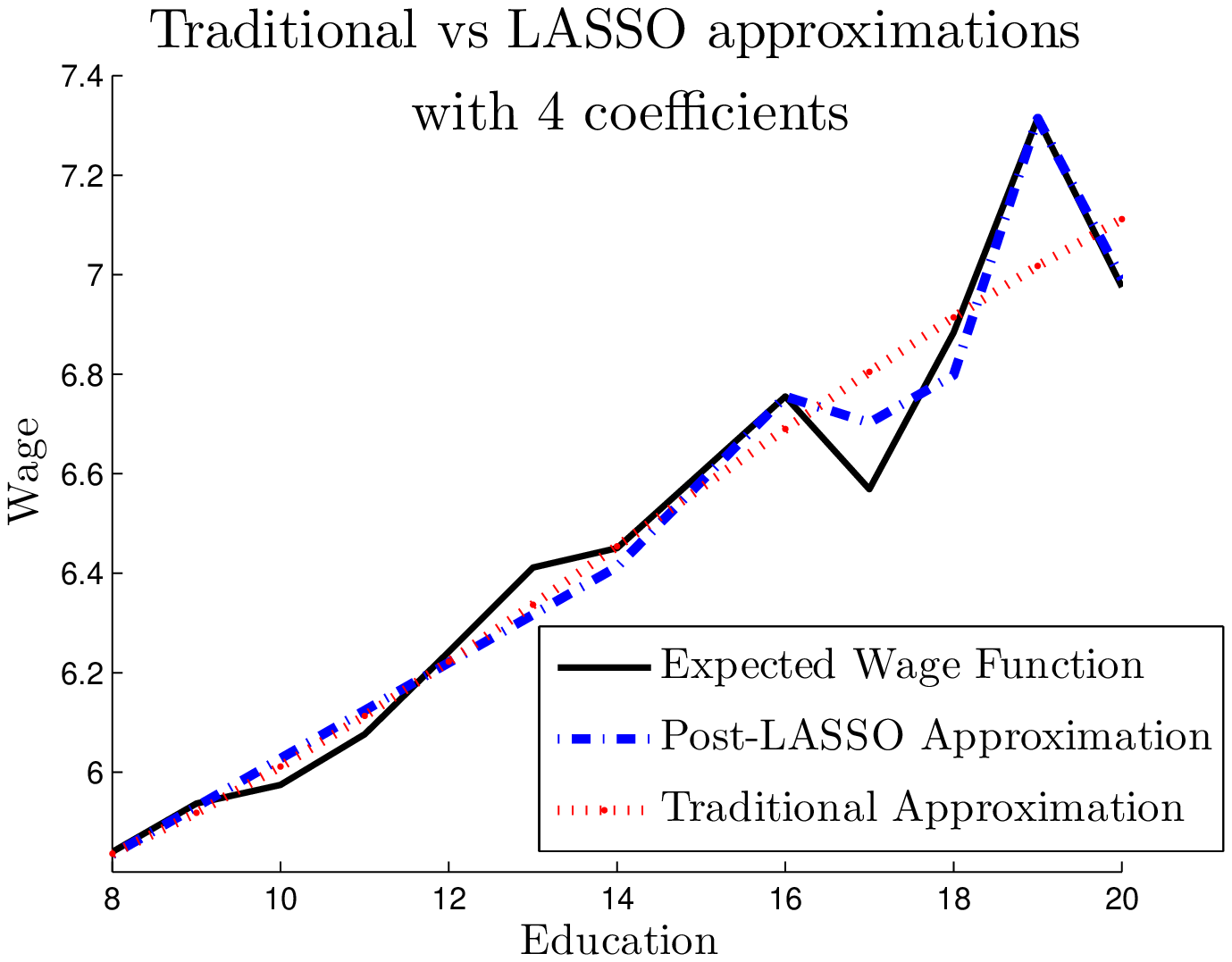} 

\includegraphics[width=0.7\textwidth]{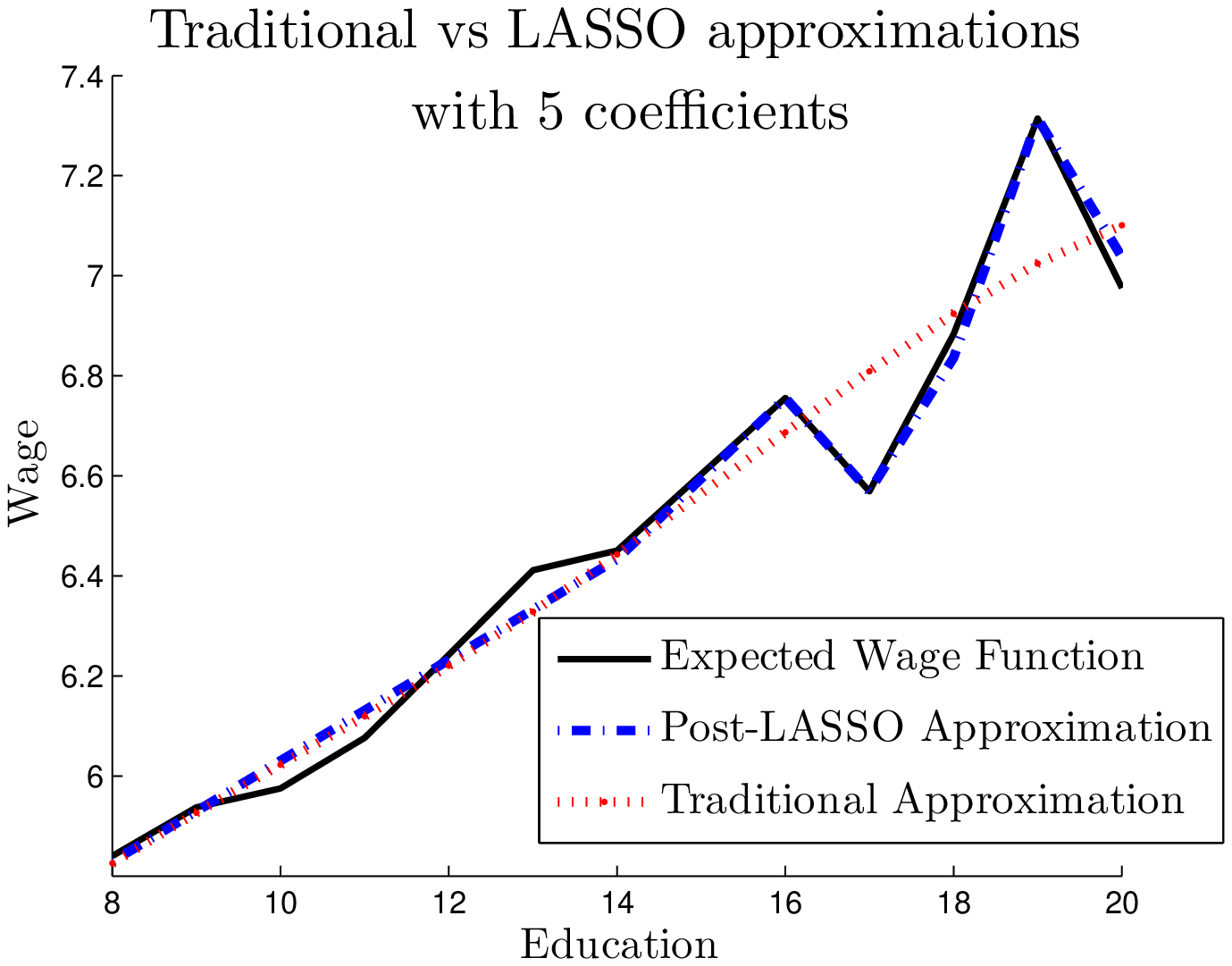} 
\caption{The figures illustrates the Post-LASSO sparse approximation and the traditional (low degree polynomial) approximation of the wage function. The top figure uses $s=4$
 and the bottom figure uses $s=5$.}\label{Fig:Wage}
\end{figure}

The next two applications are natural examples with large sets of regressors among which we need to select some smaller sets to be used in further estimation and inference.  These examples illustrate the potential wide applicability of HDS modeling in econometrics, since many classical and new data sets have naturally multi-dimensional regressors. For example, the American Housing Survey records prices and multi-dimensional features of houses sold, and scanner data-sets record prices and multi-dimensional information on products sold at a store or on the internet. \\

\textbf{Example 2: Instrument Selection in Angrist and Krueger Data}. The second example we consider is an instrumental variables model, as in Angrist and Krueger \cite{AK1991}$$
\begin{array}{lll}
 y_{i1}  & = \theta_0 + \theta_1 y_{i2} +
w_i'\gamma + v_i, &  E[v_i|w_i, x_i] = 0,\\
 y_{i2}  & = x_i'\beta + w_i'\delta +
\varepsilon_i,  &  E[\varepsilon_i|w_i, x_i] = 0,
\end{array}
 $$
where, for person $i$, $y_{i1}$ denotes wage,  $y_{i2}$ denotes education, $w_i$ denotes a vector of control variables, and $x_i$ denotes
a vector of instrumental variables that affect education but do not directly affect the wage. The instruments $x_i$ come from the quarter-of-birth dummies, and from a very large list, total of $180$, formed by interacting quarter-of-birth dummies with control variables $w_i$.   The interest focuses on measuring the coefficient $\theta_1$,
  which summarizes the causal impact
  of education on earnings, via instrumental variable estimators.

There are two basic options used in the literature: one uses just the quarter-of-birth dummies, that is, the leading
3 instruments, and another uses all 183 instruments. It is well known that using just 3 instruments results in estimates of the schooling coefficient $\theta_1$ that have a large variance and small bias,
while using 183 instruments results in estimates that have a  much smaller variance but (potentially) large bias, see, e.g., \cite{hhn:weakiv}.  It turns out that, under some conditions, by using $\ell_1$-based estimation of the first stage,
we can construct estimators that also have a nearly efficient variance and at the same time small bias.
Indeed, as shown in Table \ref{Table:IV}, using the LASSO estimator induced by different penalty levels defined in Section \ref{Sec:SettingEstimators},
it is possible to find just 37 instruments that contain nearly all information in the first stage equation. Limiting the number of the instruments from 183 to just 37 reduces the bias of the final instrumental variable estimator. For a further analysis of IV estimates based on LASSO-selected instruments, we refer the reader to \cite{BCCH-LASSOIV}.

\begin{center}
\begin{table}
\begin{center}
\caption{Instrumental Variable Estimates of Return to Schooling in Angrist and Krueger Data}

\begin{tabular}{c | c | c}
 \hline\hline
  Instruments &  Return to Schooling &  Robust Std Error  \\
\hline
   3    & 0.1077 & 0.0201\\
   180  & 0.0928 & 0.0144\\
   \hline
   LASSO-selected & & \\
   \hline
   5    & 0.1062 & 0.0179\\
   7    & 0.1034 & 0.0175\\
   17   & 0.0946 & 0.0160\\
   37   & 0.0963 & 0.0143\\
  \hline\hline
  \end{tabular}\label{Table:IV}
\end{center}
  \end{table}
\end{center}
\qed

\textbf{Example 3: Cross-country Growth Regression.} One of the central issues in the empirical growth literature is estimating the effect of an initial (lagged) level of GDP (Gross Domestic Product) per capita on the growth rates of GDP per capita. In particular, a key prediction from the classical Solow-Swan-Ramsey growth model
is the hypothesis of convergence, which states that poorer countries should typically grow faster and therefore should tend to catch up
with the richer countries. Such a hypothesis implies that the effect of the initial level of GDP on the growth rate should be negative.
As pointed out in Barro and Sala-i-Martin \cite{BarroSala1995}, this hypothesis is rejected using a simple bivariate regression of
growth rates on the initial level of GDP. (In this data set, linear regression yields an insignificant positive coefficient of $0.0013$.)
In order to reconcile the data and the theory, the literature has focused on estimating the effect \textit{conditional} on the
pertinent characteristics of countries.  Covariates that describe such characteristics can include variables measuring
education and science policies, strength of market institutions, trade openness, savings rates and others \cite{BarroSala1995}.
The theory then predicts that  for countries with similar other characteristics the effect of the initial level of GDP on the growth rate should be
negative (\cite{BarroSala1995}).  Thus, we are interested in a specification of the form:
\begin{equation}\label{growth} y_i = \alpha_0 + \alpha_1 \log G_i + \sum_{j=1}^p \beta_j X_{ij} + \varepsilon_i,
 \end{equation}
where $y_i$ is the growth rate of GDP over a specified decade in country $i$, $G_i$ is the initial level of GDP
at the beginning of the specified period, and the $X_{ij}$'s form a long list of country $i$'s characteristics at the beginning of the specified period.  We are interested in testing the hypothesis of convergence, namely that $\alpha_1 <0$.

Given that in standard data-sets, such as Barro and
Lee data \cite{BarroLee1994}, the number of covariates $p$ we can condition on is large, at least relative to the sample size $n$, covariate selection becomes a crucial issue in this analysis (\cite{OneMillion}, \cite{TwoMillion}).  In particular, previous findings came under severe criticism for relying on ad hoc procedures for covariate selection.  In fact, in some cases, all of the previous findings have been questioned (\cite{OneMillion}).   Since the number of covariates is high, there is no simple way to resolve the model selection problem using only classical tools. Indeed the number of possible lower-dimensional models is very large, although
\cite{OneMillion} and \cite{TwoMillion} attempt to search over several millions of these models.   We suggest $\ell_1$-penalization and post-$\ell_1$-penalization methods to address this important issue.  In Section \ref{Sec:Growth}, using these methods we estimate the growth model (\ref{growth}) and indeed find rather strong support for the hypothesis of convergence, thus confirming the basic implication of the Solow-Swan model. \qed

%
%
%
%
%

\textbf{Notation.} In what follows, all parameter values are indexed by the sample size $n$, but we omit the index whenever this does not cause
confusion. In making asymptotic statements, we assume that $n \to
\infty$ and $p=p_n  \to \infty$, and we also allow for $s=s_n \to
\infty$.  We use the notation $(a)_+ = \max\{a,0\}$, $a \vee b = \max\{ a, b\}$ and $a \wedge b = \min\{ a , b \}$. The $\ell_2$-norm is denoted by
$\|\cdot\|$ and the ``$\ell_0$-norm" $\|\cdot\|_0$ denotes the number of non-zero components of a vector. Given a vector $\delta \in \RR^p$, and a set of
indices $T \subset \{1,\ldots,p\}$, we denote by $\delta_T$ the vector in which $\delta_{Tj} = \delta_j$ if $j\in T$, $\delta_{Tj}=0$ if $j \notin T$. We also use standard notation in the empirical process literature, $$\En[f] = \En[f(w_i)] = \sum_{i=1}^n f(w_i)/n,$$ and we use the notation $a \lesssim b$ to denote $a \leqslant c b$ for some constant $c>0$ that does not depend on $n$; and
$a\lesssim_P b$ to denote $a=O_P(b)$. Moreover, for two random variables $X, Y$ we say that $X=_dY$ if they have the same probability distribution. We also define the prediction norm associated with the empirical Gram matrix $\En[x_ix_i']$ as $$\|\delta\|_{2,n} = \sqrt{\En[(x_i'\delta)^2]}.$$

\section{The Setting and Estimators}\label{Sec:SettingEstimators}


\subsection{The Model}

Throughout the rest of the chapter we consider the nonparametric model introduced in the previous section:
\begin{equation}\label{Def:NP}
y_i = f(z_i) + \varepsilon_i, \ \   \epsilon_i \sim N(0, \sigma^2),  \ \  i=1,\ldots,n,
\end{equation}
where $y_i$ is the outcome, $z_i$ is a vector of fixed regressors, and $\varepsilon_i$'s are i.i.d. disturbances. Define $x_i=P(z_i)$, where $P(z_i)$ is a
$p$-vector of transformations of $z_i$, including a constant,  and $f_i = f(z_i)$. For a conformable sparse vector $\beta_0$ to be defined below, we can rewrite (\ref{Def:NP}) in an approximately parametric form:
\begin{equation}\label{EqNinePrime}
y_i = x_i' \beta_0 + u_i, \ \   u_i = r_i + \varepsilon_i,   \ \  i=1,\ldots,n,
\end{equation} where $r_i := f_i - x_i'\beta_0$,  $i=1,\ldots,n,$ are approximation errors. We note that in the parametric case,  we may naturally choose $x_i'\beta_0=f_i$ so that $r_i=0$ for all $i=1,\ldots,n$.
In the nonparametric case, we shall choose $x_i'\beta_0$ as a sparse parametric model that yields a good approximation to the true regression function $f_i$ in equation (\ref{Def:NP}).

Given (\ref{EqNinePrime}), our target in estimation  will become the parametric function $x_i'\beta_0$. Here we emphasize that the ultimate target in estimation is, of course, $f_i$,
while  $x_i'\beta_0$ is a convenient intermediate target, introduced so that we can approach the estimation problem as if it were parametric. Indeed, the two targets are equal up to approximation errors $r_i$'s
that will be set smaller than estimation errors. Thus, the problem of estimating the parametric target $x_i'\beta_0$ is equivalent to
the problem of estimating  the  non-parametric target $f_i$
modulo approximation errors.

 With that in mind, we choose our target or ``true" $\beta_0$, with the corresponding cardinality of its support $$s= \| \beta_0\|_0,$$ as any solution to the following ideal risk minimization or oracle problem:
\begin{equation}\label{oracle}
 \min_{  \beta \in \RR^p }  \En [(f_i - x_i'\beta)^2] + \sigma^2 \frac{\|\beta\|_0}{n}.
\end{equation}
We call this problem the oracle problem for the reasons explained below, and  we call
$$ T = \supp(\beta_0)$$
the oracle or the ``true" model.  Note that we necessarily have that $s \leqslant n$.

The oracle problem (\ref{oracle}) balances
the approximation error $\En [(f_i - x_i'\beta)^2]$ over the design points with the variance term $\sigma^2 \|\beta\|_0/n$, where the latter is determined by the
number of non-zero coefficients in $\beta$.  Letting
$$
c^2_s:= \En[r^2_i] =  \En [(f_i - x_i'\beta_0)^2]
$$
denote the average square error from approximating values $f_i$ by $x_i'\beta_0$, the quantity $ c^2_s +  \sigma^2 s/n$ is  the optimal value of (\ref{oracle}). Typically, the
optimality in (\ref{oracle}) would balance the approximation error with the variance term so that for some absolute constant $K\geqslant 0$ \begin{equation}\label{Def:ApproxError}c_s \leqslant K \sigma \sqrt{s/n},\end{equation}
so that  $ \sqrt{c^2_s +  \sigma^2 s/n} \lesssim \sigma \sqrt{s/n}.$ Thus, the quantity $\sigma \sqrt{s/n}$ becomes the ideal goal for the rate of convergence. If we knew the oracle  model $T$, we would achieve this rate by using the oracle estimator, the least squares estimator based on this model,  but we in general do not know $T$, since we do not observe the $f_i$'s to attempt to solve the oracle problem (\ref{oracle}).  Since
$T$ is unknown, we will not be able to achieve the exact oracle rates of convergence, but we can hope to come close to this rate.

We consider the case of fixed design, namely we treat the covariate values $x_1,\ldots, x_n$ as fixed. This includes random sampling as a special case; indeed, in this case $x_1,\ldots, x_n$ represent a realization of this sample on which we condition throughout.  Without loss of generality, we normalize the covariates so that \begin{equation}\label{Def:Normalization}\text{ $\hat \sigma_j^2 = \En[x_{ij}^2] = 1$ for $j=1,\ldots,p$.}\end{equation}

We summarize the setup as the following condition.

~\\

\noindent \textbf{Condition ASM.} \textit{We have data $\{(y_i,z_i), i=1,\ldots,n\}$  that for each $n$ obey the regression model (\ref{Def:NP}), which admits the approximately sparse form  (\ref{EqNinePrime})
induced by (\ref{oracle}) with the approximation error satisfying (\ref{Def:ApproxError}). The regressors $x_i=P(z_i)$ are normalized as in (\ref{Def:Normalization}).}

~\\

\begin{remark}[On the Oracle Problem] Let us now briefly explain what is behind problem (\ref{oracle}).  Under some mild assumptions,
this problem directly arises as the (infeasible) oracle risk minimization problem. Indeed, consider an OLS estimator $\widehat\beta[\widetilde T]$, which is obtained
by using a model $\widetilde T$, i.e. by regressing $y_i$ on regressors $x_{i}[\widetilde T]$, where
$x_{i}[\widetilde T]= \{ x_{ij}, j \in \widetilde T\}$. This estimator  takes value $\widehat\beta[\widetilde T]=\En[x_{i}[\widetilde T] x_{i}[\widetilde T]']^{-} \En[x_{i}[\widetilde T]y_i]$.
The expected risk of this estimator $\En \Ep [f_i - x_i[\widetilde T]'\widehat\beta[\widetilde T]]^2$ is equal to
$$
 \min_{\beta \in \Bbb{R}^{|\widetilde T|}} \En[ (f_i - x_i[\widetilde T]' \beta)^2] +   \sigma^2 \frac{k}{n},
$$
where $k = \text{rank} (\En[x_{i}[\widetilde T] x_{i}[\widetilde T]'])$.  The oracle knows
the risk of each of the models $\widetilde T$ and can minimize this risk
$$
\min_{\widetilde T}  \min_{\beta \in \Bbb{R}^{|\widetilde T|}} \En[ (f_i - x_i[\widetilde T]' \beta)^2] +   \sigma^2 \frac{k}{n},
$$
by choosing the best model or the oracle model $T$. This problem  is in fact equivalent to (\ref{oracle}), provided that
  $\text{rank} \left(\En[x_{i}[T] x_{i}[T]']\right) =\|\beta_0\|_0$, i.e. full rank.  Thus, in this case
  the value $\beta_0$ solving (\ref{oracle}) is the expected value
of the oracle least squares estimator $\widehat\beta_{T}=\En[x_{i}[T] x_{i}[T]']^{-1} \En[x_{i}[T]y_i]$, i.e.
$\beta_0 = \En[x_{i}[T] x_{i}[T]']^{-1} \En[x_{i}[T]f_i]$.  This value is our target or ``true" parameter value
and the oracle model $T$ is the target  or ``true" model. Note that when $c_s=0$ we have that $f_i = x_i'\beta_0$, which gives
us the special parametric case.
\end{remark}


\subsection{LASSO and Post-LASSO Estimators}

Having introduced the model (\ref{EqNinePrime}) with the target parameter defined via (\ref{oracle}), our task becomes to estimate $\beta_0$.
We will focus on deriving rate of convergence results in the {\it prediction  norm}, which measures the accuracy of predicting $x_i'\beta_0$  over the design points $x_1,\ldots,x_n$,
$$
\|\delta\|_{2,n} =  \sqrt{\En[x_i'\delta]^2}.
$$ In what follows $ \delta$ will denote deviations of the estimators from the true parameter value. Thus, e.g., for $ \delta = \hat \beta - \beta_0$, the quantity $\|\delta\|_{2,n}^2$ denotes the average of the square errors $x_i'\hat\beta - x_i'\beta_0$ resulting from using the estimate $x_i'\hat\beta$ instead of $x_i'\beta_0$.
Note that once we bound $\widehat \beta - \beta_0$ in the prediction norm, we can also bound the empirical risk of predicting
values $f_i$ by $x_i'\widehat \beta$  via the triangle inequality:
\begin{equation}\label{Def:NORM_ER}
\sqrt{\En[ (x_i'\hat \beta - f_i)^2]} \leqslant  \|\widehat \beta - \beta_0\|_{2,n}+c_s.
\end{equation}

In order to discuss estimation consider first the classical ideal AIC/BIC type estimator (\cite{Akaike1974,Schwarz1978}) that solves the empirical (feasible) analog of the oracle problem:
$$
\min_{\beta \in \Bbb{R}^p} \widehat Q (\beta) + \frac{\lambda}{n}  \|\beta\|_{0},
$$
where  $ \widehat Q (\beta) = \En[(y_i - x_i'\beta)^2]$ and  $
 \| \beta\|_0 = \sum_{j=1}^p 1\{ | \beta_j| > 0 \}$ is the $\ell_0$-norm and $\lambda$ is the penalty level.  This estimator has very attractive theoretical properties, but unfortunately it is computationally prohibitive, since the solution to the problem may require solving $\sum_{k \leqslant n} \binom{p}{k}$ least squares problems (generically, the complexity of this problem is NP-hard \cite{Natarajan1995,GeJiangYe2010}).

One way to overcome the computational difficulty is to consider a convex relaxation of the preceding problem, namely
to employ the closest convex penalty -- the $\ell_1$ penalty -- in place of the $\ell_0$ penalty. This construction leads to the so called LASSO
estimator:\footnote{The abbreviation LASSO stands for Least Absolute Shrinkage and Selection Operator, c.f.  \cite{T1996}.}
\begin{equation}\label{Def:LASSOmain}
\widehat \beta \in \arg \min_{\beta \in \Bbb{R}^p} \widehat Q (\beta) + \frac{\lambda}{n} \| \beta \|_{1},
\end{equation}
where as before $ \widehat Q (\beta) = \En[(y_i - x_i'\beta)^2]$ and $\|\beta\|_{1} = \sum_{j=1}^p | \beta_j|$. The LASSO estimator
minimizes a convex function. Therefore, from a computational complexity perspective, (\ref{Def:LASSOmain}) is a computationally efficient (i.e. solvable in polynomial time) alternative to
AIC/BIC estimator.

In order to describe the choice of $\lambda$, we highlight that the following key quantity determining this choice:
$$ S = 2\En[x_i\varepsilon_i], $$
which summarizes the noise in the problem. We would like to choose the smaller penalty level so that
\begin{equation}\label{choice of lambda probability}
 \lambda \geqslant  c n \|S\|_{\infty} \text{ with probability at least } 1- \alpha,
 \end{equation}
where $1-\alpha$ needs to be close to one, and $c$ is a constant such that $c>1$.
Following \cite{BC-PostLASSO} and
\cite{BickelRitovTsybakov2009},
 respectively, we consider two choices of $\lambda$ that achieve the above:
\begin{eqnarray}
& \label{Def:Xindependent} X\mbox{-independent penalty:} \quad & \lambda := 2c\sigma\sqrt{n} \Phi^{-1}(1-\alpha/2p), \\
& \label{Def:Xdependent}  X\mbox{-dependent penalty:} \quad & \lambda :=  2c\sigma \Lambda(1-\alpha|X),
\end{eqnarray}
where $\alpha \in (0,1)$ and $c>1$ is constant, and
$$
\Lambda(1-\alpha|X) := (1-\alpha)-\text{quantile of }  n\|S/(2\sigma)\|_{\infty} ,
$$
$ \text{ conditional on } X=(x_1,\ldots,x_n)'$.
Note that
$$
\|S/(2\sigma)\|_{\infty}  =_d \max_{1 \leqslant j \leqslant p}  |\En [x_{ij} g_i]|, \text{ where $g_i$'s are i.i.d. } N(0,1),
$$
conditional on $X$, so we can compute $\Lambda(1-\alpha|X)$ simply by simulating the latter quantity, given the fixed design matrix $X$.
Regarding the choice of $\alpha$ and $c$,  asymptotically we require $\alpha \to 0$ as $n \to \infty$ and  $c>1$. Non-asymptotically, in our finite-sample experiments, $\alpha = .1$ and $c=1.1$ work quite well.
The noise level $\sigma$ is unknown in practice, but we can estimate it consistently using the approach of Section 6.
  We recommend the $X$-dependent rule over the $X$-independent rule, since the former by construction adapts to the design matrix $X$ and is less conservative than the latter in view
of the following relationship that follows from Lemma \ref{Lemma:GaussianTail}:
\begin{equation}\label{ranking}
\Lambda(1-\alpha|X) \leqslant \sqrt{n} \Phi^{-1}(1-\alpha/2p) \leqslant \sqrt{2 n \log(2p/\alpha)}.
\end{equation}

Regularization by the $\ell_1$-norm employed in (\ref{Def:LASSOmain}) naturally helps the LASSO estimator to avoid overfitting the data, but it also shrinks the fitted coefficients towards zero, causing a potentially significant bias. In order to remove some of this bias, let us consider the Post-LASSO estimator that applies ordinary least squares regression to  the model $\widehat T$ selected by LASSO. Formally, set
$$\widehat T = \supp( \hat \beta ) = \{ j \in \{1,\ldots,p\} \ : \ |\hat\beta_j| > 0\},$$
and define the Post-LASSO estimator $\widetilde \beta$ as \begin{equation}\label{Def:TwoStep} \widetilde \beta \in \arg\min_{\beta \in \mathbb{R}^p} \ \widehat Q(\beta) : \beta_j = 0 \text{ for each } j \in \widehat T^c,
\end{equation}
where $\widehat T^c = \{1,...,p\} \setminus \widehat T$. In  words, the estimator is ordinary least squares applied to the data after removing the regressors that were not selected by LASSO. If the model selection works perfectly -- that is, $\widehat T = T$ --
then the Post-LASSO estimator is simply the oracle estimator whose properties are well known. However, perfect model selection might be unlikely for many designs of interest, so we are especially interested in the properties of Post-LASSO in such cases, namely when $\widehat T \neq T$, especially when $T\nsubseteq \widehat T$.

\subsection{Intuition and Geometry of LASSO and Post-LASSO}

In this section we discuss the intuition behind LASSO and Post-LASSO estimators defined above.
We shall rely on a dual interpretation of the LASSO optimization problem to provide some geometrical intuition for the performance of LASSO.
Indeed, it can be seen that the LASSO estimator also solves the following optimization program:
\begin{equation}\label{dual} \min_{\beta \in \RR^p} \|\beta\|_1 : \ \widehat Q(\beta) \leqslant \gamma
\end{equation}
for some value of $\gamma \geqslant 0$ (that depends on the penalty level $\lambda$).  Thus, the estimator minimizes the $\ell_1$-norm
of coefficients subject to maintaining a certain goodness-of-fit; or, geometrically,
the LASSO estimator searches for a minimal $\ell_1$-ball -- the diamond-- subject to the diamond having a non-empty intersection with a fixed lower contour set of the least squares criterion function -- the ellipse.

In Figure \ref{Fig:PLOTS}  we show an illustration for the two-dimensional case with the true parameter value $(\beta_{01}, \beta_{02})$ equal $(1,0)$, so that $T=\text{support}(\beta_0)=\{1\}$ and $s=1$.  In the figure we plot the diamonds and ellipses. In the top figure, the ellipse represents a lower contour set of the population criterion function $Q(\beta) = \Ep[(y_i-x_i'\beta)^2]$ in the zero noise case or the infinite sample case.  In the bottom figures the ellipse represents a contour set of the sample criterion function  $\hat
Q(\beta)=\En[(y_i-x_i'\beta)^2]$ in the non-zero noise or the finite sample case.
 The set of optimal solutions $\widehat \beta$ for LASSO is then given by the intersection of the minimal diamonds with the ellipses.  Finally, recall that Post-LASSO is computed as the ordinary least square solution using covariates selected by LASSO. Thus, Post-LASSO estimate $\widetilde \beta$ is given by the center of the ellipse intersected with the linear subspace selected by LASSO.

In the zero-noise case or in population (top figure), LASSO easily recovers the correct sparsity pattern of $\beta_0$. Note that due to the regularization, in spite of the absence of noise, the LASSO estimator has a large bias towards zero. However, in this case Post-LASSO $\widetilde \beta$ removes the bias and recovers $\beta_0$ perfectly.

In the non-zero noise case (middle and bottom figures), the contours of the criterion function and its center move away from the population counterpart. The empirical error in the middle figure moves the center of the ellipse to a non-sparse point. However, LASSO correctly sets $\widehat \beta_2 = 0$ and $\widehat\beta_1\neq 0$ recovering the sparsity pattern of $\beta_0$. Using the selected support, Post-LASSO $\widetilde \beta$ becomes the oracle estimator which drastically improves upon LASSO. In the case of the bottom figure, we have large empirical errors that push the center of the lower contour set further away from the population counterpart. These large empirical errors make the LASSO estimator non-sparse, incorrectly setting $\widehat \beta_2 \neq 0$. Therefore, Post-LASSO uses  $\widehat T = \{1,2\}$ and does not use the exact support $T=\{1\}$. Thus, Post-LASSO is not the oracle estimator in this case.

All three figures also illustrate the shrinkage bias towards zero in the LASSO estimator that is introduced by the $\ell_1$-norm penalty. The Post-LASSO estimator is motivated as a solution to remove (or at least alleviate) this shrinkage bias.
In cases where LASSO achieves a good sparsity pattern, Post-LASSO can drastically improve upon LASSO.

\begin{figure}[h!]
\centering
\includegraphics[width=0.79\textwidth]{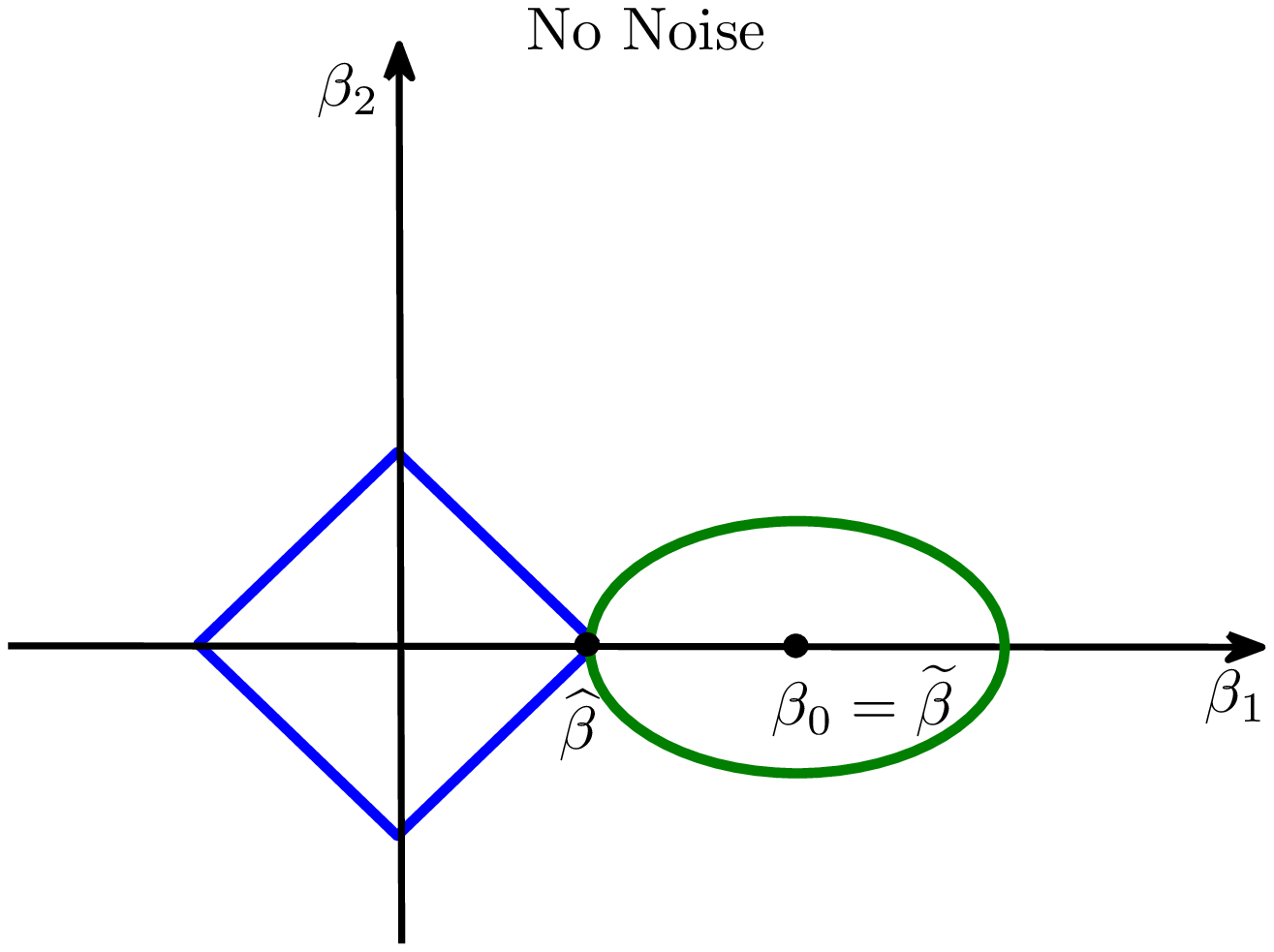}

\includegraphics[width=0.79\textwidth]{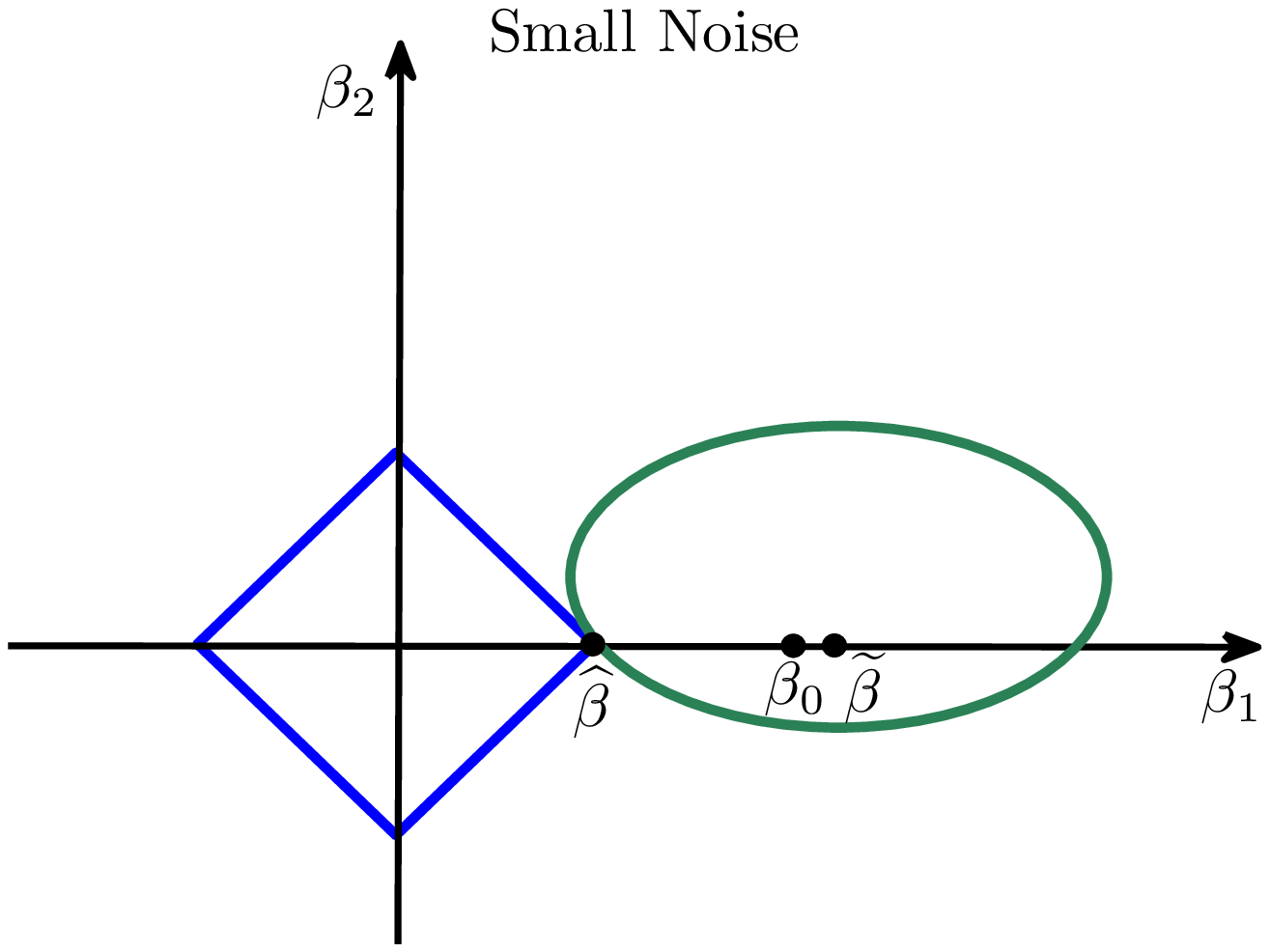}
\includegraphics[width=0.79\textwidth]{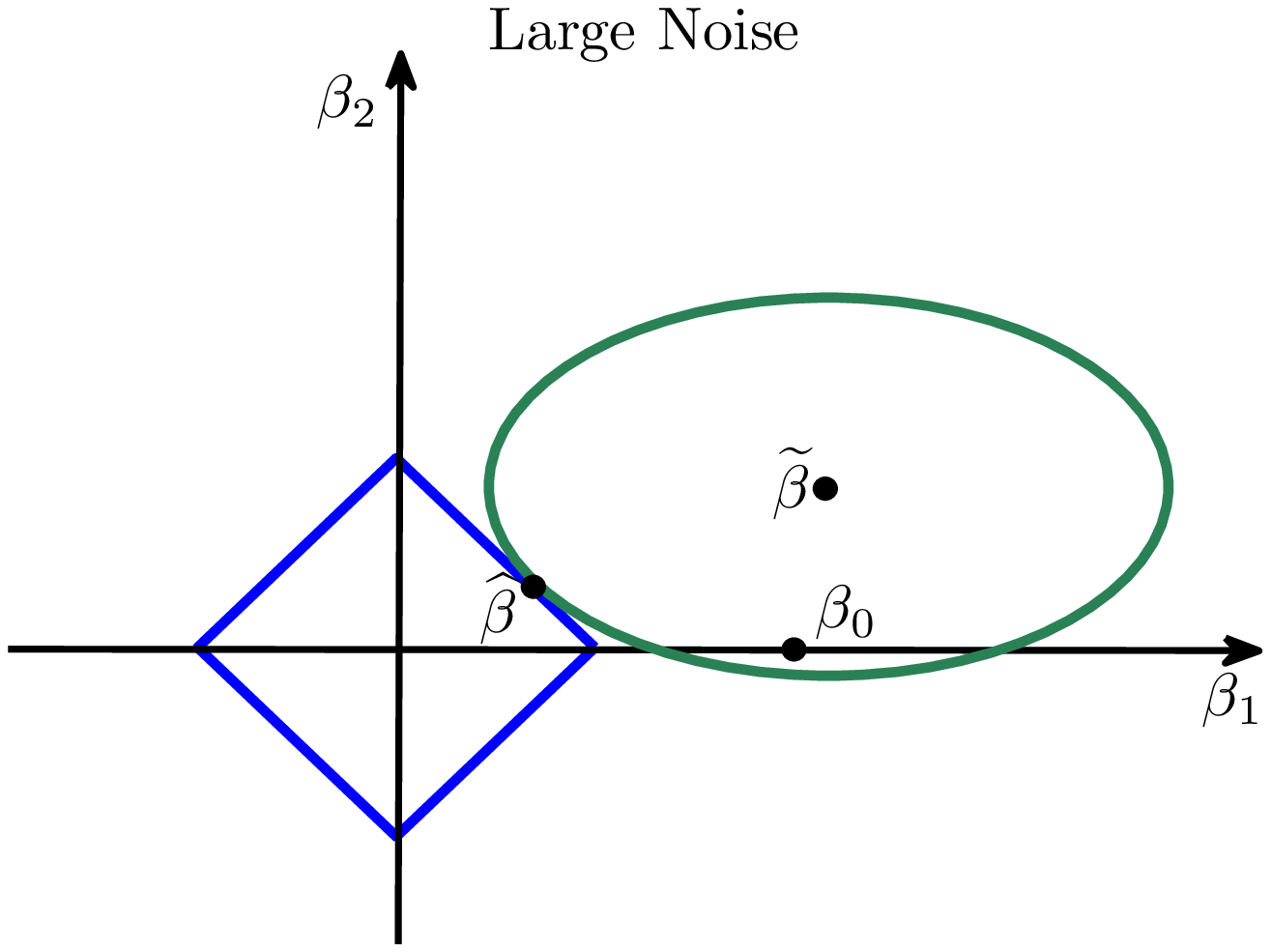}
\caption{The figures illustrate the geometry of LASSO and Post-LASSO estimator.}\label{Fig:PLOTS}
\end{figure}

\subsection{Primitive conditions}\label{Sec:Primitive}

In both the parametric and non-parametric models described above, whenever $p>n$, the empirical Gram matrix $\En[x_ix_i']$ does not
have full rank and hence it is not well-behaved. However, we only need good behavior of certain moduli of continuity of the Gram matrix called restricted sparse eigenvalues.
We define the minimal restricted sparse eigenvalue
\begin{equation}\label{Def:RSE1}
\kappa(m)^2 : = \min_{\|\delta_{T^c}\|_{0} \leqslant m, \delta \neq 0
 } \frac{ \|\delta\|_{2,n}^2}{\|\delta\|^2},
\end{equation}
and the maximal restricted sparse eigenvalue as
\begin{equation}\label{Def:RSE2}
 \phi(m) : = \max_{\|\delta_{T^c}\|_{0} \leqslant m, \delta \neq 0
 } \frac{ \|\delta\|^2_{2,n}}{\|\delta\|^2},
\end{equation} where $m$ is the upper bound on the number of non-zero components outside the support $T$.
To assume that $\kappa(m)>0$ requires that all empirical Gram submatrices formed by any $m$ components of $x_i$ in addition to the components in $T$ are positive definite. It
will be convenient to define the following sparse \textit{condition} number associated with the empirical Gram matrix:
\begin{equation}\label{Def:mmu}
\mmu{m} = \frac{\sqrt{\phi(m)}}{\kappa(m)}.
\end{equation}

In order to state simplified asymptotic statements, we shall also invoke the following condition.
~\\

\noindent \textbf{Condition RSE.} \textit{Sparse eigenvalues of the empirical Gram matrix are well behaved, in the sense that for $m=m_n=s\log n$}
\begin{equation}\label{Eq:PRIMITIVE}
\mmu{m} \lesssim 1,  \quad \phi(m) \lesssim 1, \quad   1/\kappa(m) \lesssim 1.
\end{equation}

This condition holds with high probability for many designs of interest under mild conditions on $s$. For example, as shown in Lemma \ref{Lemma:GaussianDesign}, when the covariates are Gaussians, the conditions in (\ref{Eq:PRIMITIVE}) are true with probability converging to one under the mild assumption that $s\log p = o(n)$. Condition RSE is likely to hold for other regressors with jointly light-tailed distributions, for instance log-concave distribution. As shown in Lemma \ref{Lemma:BoundedDesign}, the conditions in (\ref{Eq:PRIMITIVE}) also hold for general bounded regressors under the assumption that $s (\log^4 n)\log (p\vee n) = o(n)$. Arbitrary bounded regressors often arise in non-parametric models, where
 regressors $x_i$ are formed as spline, trigonometric, or polynomial transformations $P(z_i)$ of some elementary bounded regressors $z_i$.

\begin{lemma}[Gaussian design]\label{Lemma:GaussianDesign}   Suppose $\tilde x_i$, $i = 1,\ldots,n$, are i.i.d. zero-mean Gaussian random vectors,
such that the population design matrix $\Ep[\tilde x_i \tilde x'_i]$ has ones on the diagonal, and its $s\log n$-sparse eigenvalues are bounded from above by $\varphi
< \infty$ and bounded from below by $\kappa^2>0$.  Define $x_{i}$ as a normalized form of $\tilde x_i$, namely   $x_{ij}= \tilde
x_{ij}/\sqrt{\En[\tilde x_{ij}^2]}$.  Then for any $m\leqslant (s\log(n/e)) \wedge (n/[16\log p])$, with probability at least $1-2\exp(-n/16)$,
$$  \phi(m) \leqslant 8\varphi, \ \  \kappa(m) \geqslant \kappa/6\sqrt{2}, \ \ \mbox{and} \ \ \mmu{m} \leqslant 24\sqrt{\varphi}/\kappa. $$
\end{lemma}

\begin{lemma}[Bounded design]\label{Lemma:BoundedDesign} Suppose $\tilde x_i$, $i = 1,\ldots,n$, are i.i.d. vectors, such that the population design matrix $\Ep[\tilde x_i \tilde x'_i]$ has ones on the diagonal, and its $s\log n$-sparse eigenvalues are bounded from above by $\varphi < \infty$ and bounded from below by $\kappa^2>0$.  Define $x_{i}$ as a normalized form of $\tilde x_i$, namely   $x_{ij}= \tilde x_{ij}/(\En[\tilde x_{i j}^2])^{1/2}$. Suppose that    $\max_{1\leq i\leq n}\|\tilde x_{i}\|_\infty \leq K_n$ a.s., and $K^2_ns\log^2 (n)  \log^2 (s\log n)  \log (p\vee n) =  o(n \kappa^4/\varphi)$. Then, for any $m\geq 0$ such that $m+s\leq  s\log n$, we have that as $n \to \infty$
$$ \phi(m) \leqslant 4\varphi, \ \  \kappa(m) \geqslant \kappa/2, \ \ \mbox{and} \ \ \mmu{m} \leqslant 4\sqrt{\varphi}/\kappa, $$
with probability approaching 1.
\end{lemma}

For  proofs, see \cite{BC-PostLASSO}; the first lemma builds upon results in \cite{ZhangHuang2006} and the second builds upon results in \cite{RudelsonVershynin2008}.

\section{Analysis of LASSO}\label{Sec:LASSO}

In this section we discuss the rate of convergence of LASSO in the prediction norm;  our exposition follows
mainly \cite{BickelRitovTsybakov2009}.

 The key quantity in the analysis is the following quantity called ``score":
$$
S = S(\beta_0) = 2 \En [ x_i \varepsilon_i].
$$
The score is  the effective ``noise" in the problem. Indeed, defining $
\delta := \widehat \beta - \beta_0,
$
note that by the H\"older's inequality
 \begin{equation}
\begin{array}{rcl} \hat Q(\hat \beta) - \hat Q(\beta_0) - \| \delta\|^2_{2,n} &  = &   -2\En[\varepsilon_i x_i' \delta] - 2\En[r_i x_i' \delta] \\
& \geqslant & - \| S\|_{\infty} \| \delta\|_{1} - 2c_s\|\delta\|_{2,n}.
  \end{array}
 \end{equation}
Intuition suggests that we need to majorize the ``noise term" $\| S\|_{\infty}$ by the penalty level $\lambda/n$, so that the bound on $\| \delta\|^2_{2,n} $ will follow from a relation between the prediction norm $\| \cdot \|_{2,n}$ and the penalization norm $\|\cdot\|_{1}$ on a suitable set. Specifically, for any $c > 1$, it will follow that if $$\lambda \geqslant cn\|S\|_\infty$$ and $\|\delta\|_{2,n} \geqslant 2c_s$, the vector $\delta$ will also satisfy \begin{equation}\label{Def:Dominant}\|\delta_{T^c}\|_1\leqslant \cc\|\delta_T\|_1,\end{equation}
where $\cc = (c+1)/(c-1)$. That is, in this case the error in the regularization norm outside the true support does not exceed $\cc$ times the error in the true support.  (In the case $\|\delta\|_{2,n} \leqslant 2c_s$ the inequality (\ref{Def:Dominant}) may not hold, but the bound $\|\delta\|_{2,n} \leqslant 2c_s$ is already good enough.)

Consequently, the analysis of the rate of convergence of LASSO relies on the so-called restricted eigenvalue $\kappa_\cc$, introduced in \cite{BickelRitovTsybakov2009}, which controls the modulus of continuity between the prediction norm $\| \cdot \|_{2,n}$ and the penalization norm $\|\cdot\|_{1}$ over the set of vectors $\delta \in \Bbb{R}^p$ that satisfy (\ref{Def:Dominant}):
\begin{equation}\label{Def:RE1}\tag{
\text{RE($c$)}} \ \ \hfill  \kappa_\cc : = \min_{\|\delta_{T^c}\|_{1} \leqslant \cc  \| \delta_{T}\|_{1}, \delta_T \neq 0
 } \frac{\sqrt{s}\|\delta\|_{2,n}}{\|\delta_T\|_{1} },\end{equation}
where $\kappa_\cc$ can depend on $n$.
The constant $\kappa_\cc$ is a crucial technical quantity in our analysis and we need to bound it away from zero. In the leading cases that condition RSE holds this will in fact be the case as the sample size grows, namely
\begin{equation}\label{eq:PRIMITIVE2} 1/\kappa_\cc \lesssim 1.
\end{equation}

Indeed, we can bound $\kappa_\cc$ from below by
$$\begin{array}{rcl}
 \kappa_\cc & \geqslant & \displaystyle \max_{m\geqslant 0} \kappa(m) \( 1 -  \mmu{m} \cc\sqrt{s/m}\)
  \geqslant  \kappa(s\log n) \( 1 -  \mmu{s\log n} \cc\sqrt{1/\log n}\) \end{array}$$
by Lemma \ref{Lemma:BoundKAPPA} stated and proved in the appendix. Thus, under the condition RSE, as $n$ grows, $\kappa_\cc$ is bounded away from zero since $\kappa(s\log n)$ is bounded away from zero and $\phi(s\log n)$ is bounded from above as in (\ref{Eq:PRIMITIVE}). Several other primitive assumptions can be used to bound $\kappa_\cc$. We refer the reader to \cite{BickelRitovTsybakov2009} for a further detailed discussion of lower bounds on $\kappa_\cc$.

We next state a non-asymptotic performance bound  for the LASSO estimator.

\begin{theorem}[Non-Asymptotic Bound for LASSO]\label{Thm:Nonparametric}  Under condition ASM, the event $\lambda \geqslant c n\| S\|_{\infty}$ implies
\begin{eqnarray}\label{eq: finite sample bound}
&& \|\widehat \beta - \beta_0\|_{2,n}  \leqslant \(1 + \frac{1}{c}\) \frac{\lambda \sqrt{s}}{n \kappa_\cc} + 2 c_s,
 \end{eqnarray}
 where $c_s = 0$ in the parametric case, and $\cc = (c+1)/(c-1)$.  Thus, if $\lambda \geqslant c n\| S\|_{\infty}$ with probability
 at least $1-\alpha$, as guaranteed by either $X$-independent or $X$-dependent penalty levels (\ref{Def:Xindependent}) and (\ref{Def:Xindependent}), then the bound (\ref{eq: finite sample bound}) occurs
 with probability at least $1-\alpha$.
  \end{theorem}
The proof of Theorem \ref{Thm:Nonparametric} is given in the appendix.
The theorem  also leads to the following useful asymptotic bounds.

\begin{corollary}[Asymptotic Bound for LASSO]\label{corollary1:rate}
Suppose that conditions ASM and RSE hold. If $\lambda$ is chosen according to either the $X$-independent or $X$-dependent rule specified in  (\ref{Def:Xindependent}) and (\ref{Def:Xdependent}) with $\alpha = o(1)$, $\log(1/\alpha) \lesssim \log p$, or more generally so that
\begin{equation}\label{valid lambda}
\lambda \lesssim_P \sigma\sqrt{n\log p} \ \text{ and } \  \lambda \geqslant c'n\|S\|_\infty  \text { wp  } \to 1,
\end{equation}
for some $c'>1$,  then the following asymptotic bound holds:
 $$\|\widehat \beta -\beta_0 \|_{2,n} \lesssim_P \sigma \sqrt{\frac{s\log p}{n}} + c_s. $$
\end{corollary}

The non-asymptotic and asymptotic bounds for the empirical risk immediately follow from the triangle inequality:
\begin{equation}\label{Def:NORM_ER_OLD}
\sqrt{\En[ (f_i - x_i'\hat \beta)^2]} \leqslant  \|\widehat \beta - \beta_0\|_{2,n}+c_s.
\end{equation}
Thus, the rate of convergence of $x_i'\widehat \beta$ to $f_i$ coincides with the rate of convergence of the oracle estimator $\sqrt{ c^2_s +  \sigma^2 s/n}$ up to a logarithmic factor of $p$. Nonetheless, the performance of LASSO can be considered optimal in the sense that under general conditions the oracle rate is achievable only up to logarithmic factor of $p$ (see Donoho and Johnstone \cite{DonohoJohnstone1994} and Rigollet and Tsybakov \cite{RigolletTsybakov2010}), apart from very exceptional, stringent cases, in which it is possible to perform perfect or near-perfect model selection.

\section{Model Selection Properties and Sparsity of LASSO}\label{Sec:ModelSelection}

The purpose of this section is, first, to provide bounds (sparsity bounds) on the
dimension of the model selected by LASSO, and, second,
to describe some special cases where the model selected by LASSO perfectly matches the ``true" (oracle) model.

\subsection{Sparsity Bounds}

Although perfect model selection can be seen as unlikely in many designs, sparsity of the LASSO estimator has been documented in a variety of designs. Here we describe the sparsity results obtained in \cite{BC-PostLASSO}.
Let us define
$$\widehat m := |\widehat T\setminus T| = \|\widehat \beta_{T^c}\|_0,$$
which is the number of unnecessary components or regressors selected by LASSO.

\begin{theorem}[Non-Asymptotic Sparsity Bound for LASSO]\label{Thm:Sparsity}
Suppose condition ASM holds. The event $\lambda \geqslant cn\|S\|_\infty$ implies that
$$ \hat m \leqslant s \cdot \left[  \min_{m \in \mathcal{M}}\phi(m\wedge n) \right] \cdot L,$$
where  $\mathcal{M}=\{ m \in \mathbb{N}:
m > s  \phi(m\wedge n)\cdot 2L \}$ and $L = [ 2\cc/\kappa_\cc + 3(\cc+1)nc_s/(\lambda\sqrt{s})]^2$.
\end{theorem}

Under Conditions ASM and RSE, for $n$ sufficiently large we have $1/\kappa_\cc \lesssim 1$, $c_s \lesssim \sigma\sqrt{s/n}$, and
 $\phi(s\log n) \lesssim 1$; and under the conditions of Corollary \ref{corollary1:rate}, $\lambda \geq c \sigma\sqrt{n}$ with probability
 approaching one. Therefore, we  have  that $L\lesssim_P 1$ and
$$ s\log n > s \phi(s\log n) \cdot 2L, \ \ \text{that is,}  \ \  s\log n \in \mathcal{M} $$
with probability approaching one as $n$ grows.
Therefore, under these conditions we have
$$  \min_{m \in \mathcal{M}}\phi(m\wedge n)\lesssim_P 1.$$

\begin{corollary}[Asymptotic Sparsity Bound for LASSO]\label{corollary2:sparsity}
Under the conditions of Corollary \ref{corollary1:rate},  we have that
\begin{equation}\label{eq: oracle sparsity}
\hat m \lesssim_P s.
 \end{equation}
\end{corollary}

Thus, using a penalty level that satisfies (\ref{valid lambda})  LASSO's sparsity is asymptotically of the same order as the oracle sparsity, namely
\begin{equation}
\widehat s :=|\widehat T| \leqslant s + \widehat m \lesssim_P s.\end{equation}
We note here that Theorem \ref{Thm:Sparsity} is particularly helpful in designs in which
$ \min_{m \in \mathcal{M}}\phi(m)$ $\ll$ $\phi(n) $. This allows Theorem \ref{Thm:Sparsity} to sharpen the
 sparsity bound of the form $\widehat s \lesssim_P s \phi(n)$ considered  in \cite{BickelRitovTsybakov2009} and \cite{MY2007}.  The bound above is comparable to the bounds in \cite{ZhangHuang2006} in terms of order of magnitude, but Theorem \ref{Thm:Sparsity} requires a smaller penalty level $\lambda$ which also does not depend on the unknown sparse eigenvalues as in \cite{ZhangHuang2006}.


\subsection{Perfect Model Selection Results}\label{Sec:PerfectMS}

The purpose of this section is to describe very special cases
where perfect model selection is possible. Most results in the literature for model selection have been developed for the parametric case only (\cite{MY2007},\cite{Lounici2008}). Below we provide some results for the nonparametric models, which cover the parametric models
as a special case.

\begin{lemma}[Cases with Perfect Model Selection by Thresholded LASSO]\label{Lemma:Crack} Suppose condition ASM holds.
(1) If the non-zero coefficients of the oracle model are well separated from zero, that is
$$ \min_{j\in T} |\beta_{0j}| > \zeta + t, \ \ \ \text{ for some }  t\geqslant  \zeta := \max_{j=1,\ldots,p} |\widehat\beta_j - \beta_{0j}|,$$
then the oracle model is a subset of the selected model,
$$
T:=\supp(\beta_0) \subseteq \widehat T := \supp(\widehat\beta).$$
Moreover the oracle model $T$ can be perfectly selected by applying hard-thresholding of level $t$ to the LASSO estimator $\widehat \beta$:
$$
 T = \left\{ j \in \{1,\ldots, p\} \ : \ |\hat\beta_j| > t \right\}.$$
(2) In particular,  if $\lambda \geqslant c n\|S\|_{\infty}$, then for $\widehat m = |\widehat T\setminus T|=\|\widehat \beta_{T^c}\|_0$ we have
$$ \zeta \leqslant \(1 + \frac{1}{c}\) \frac{\lambda \sqrt{s}}{n \kappa_\cc \kappa(\widehat m)} + \frac{2c_s}{\kappa(\widehat m)}.
$$
 (3) In particular,  if $\lambda \geqslant c n\|S\|_{\infty}$, and there is a constant $U>5\cc$ such that the empirical Gram matrix satisfies $|\En[x_{ij}x_{ik}]|\leqslant 1/[Us]$ for all $1\leqslant j< k\leqslant p$, then



$$\zeta \leqslant \frac{\lambda}{n} \cdot \frac{U + \cc}{U - 5\cc} + \min\left\{\frac{\sigma}{\sqrt{n}}, c_s\right\} +
 \frac{6\cc}{U-5\cc} \frac{c_s}{\sqrt{s}}  + \frac{4\cc}{U}\frac{n}{\lambda} \frac{c_s^2}{s}.$$

\end{lemma}

Thus, we see from parts (1) and (2) that perfect model selection is possible under strong assumptions on the coefficients' separation away from zero.
We also see from part (3) that the strong separation of coefficients can be considerably weakened in exchange for a strong assumption on the maximal pairwise correlation of regressors. These results generalize
to the nonparametric case the results of \cite{Lounici2008} and \cite{MY2007} for the parametric case in which $c_s = 0$.

Finally, the following result on perfect model selection also requires strong assumptions
on separation of coefficients and the empirical Gram matrix. Recall that for a scalar $v$, $\sign(v)=v/|v|$ if $|v|>0$, and $0$ otherwise. If $v$ is a vector, we apply the definition componentwise. Also, given a vector $x \in \RR^p$ and a set $T  \subset \{1,...,p\}$, let
us denote $x_{i}[T]:=\{x_{ij}, j \in T\}$.


\begin{lemma}[Cases with Perfect Model Selection by  LASSO]\label{Lemma:Crack2}
Suppose condition ASM holds. We have perfect model selection for LASSO, $\widehat T = T$, if and only if
\begin{eqnarray*}\label{Def:FirstM}
& \Big \| \En\[x_{i}[T^c]x_{i}[T]'\] \En\[x_{i}[T]x_{i}[T]'\]^{-1} \Big
\{ \En[x_{i}[T]u_i] \\
&   \quad \quad - \frac{\lambda}{2n} \sign(\beta_{0}[T]) \Big \} - \En[ x_{i}[T^c]u_i] \Big \|_\infty \leqslant \frac{\lambda}{2n}, \\
&\label{Def:secondM}\min_{j \in T} \left| \beta_{0j} + \left( \En\[x_{i}[T]x_{i}[T]'\]^{-1} \left\{ \En[x_{i}[T]u_i] - \frac{\lambda}{2n} \sign(\beta_{0}[T]) \right\} \right)_j\right| > 0. \end{eqnarray*}
\end{lemma}

The result follows immediately from the first order optimality conditions, see \cite{Wainright2006}. \cite{ZhaoYu2006} and \cite{CandesPlan2009} provides
further primitive sufficient conditions for perfect model selection for the parametric case in which $u_i = \varepsilon_i$. The conditions above might typically require a slightly larger choice of $\lambda$ than (\ref{Def:Xindependent}) and larger separation from zero of the minimal non-zero coefficient $\min_{j\in T}|\beta_{0j}|$.

\section{Analysis of  Post-LASSO}\label{Sec:PostLASSO}

Next we study the rate of convergence of the Post-LASSO estimator. Recall that for $\widehat T = \text{ support } (\widehat \beta)$, the Post-LASSO estimator solves
$$
\widetilde \beta \in \arg\min_{\beta \in \mathbb{R}^p} \ \widehat Q(\beta) : \beta_j = 0 \text{ for each } j \in \widehat T^c.
$$
It is clear that if the model selection works perfectly (as it will under some rather stringent conditions discussed in Section \ref{Sec:PerfectMS}), that is, $T = \widehat T$,
then this estimator is simply the oracle least squares estimator whose properties are well known. However, if the model selection does not work perfectly, that is, $T \not = \widehat T$, the resulting performance of the estimator faces two different perils: First, in the case where LASSO selects a model $\widehat T$ that does not fully include the true model $T$, we have a specification error in the second step. Second, if LASSO includes additional regressors outside $T$, these regressors were not chosen at random and are likely to be spuriously correlated with the disturbances, so we have a data-snooping bias in the second step.

It turns out that despite of the possible poor selection of the model, and the aforementioned perils this causes, the Post-LASSO estimator still performs well theoretically, as shown in \cite{BC-PostLASSO}. Here we provide a proof similar to \cite{BCCH-LASSOIV} which is easier generalize to non-Gaussian cases.

\begin{theorem}[Non-Asymptotic Bound for Post-LASSO]\label{Cor:2StepNonparametric} Suppose condition ASM holds. If $\lambda \geqslant c n\| S\|_{\infty} $
holds with probability at least $1-\alpha$,
 then for any $\gamma > 0$ there is a constant $K_\gamma$ independent of $n$ such that with probability at least $1- \alpha - \gamma$
$$\begin{array}{rl}
 \displaystyle \|\widetilde \beta - \beta_0\|_{2,n} & \leqslant  \displaystyle
\frac{K_{\gamma}\sigma}{\kappa(\widehat m)} \sqrt{\frac{s + \widehat m \log p}{n}} + 2c_s + 1\{T \not\subseteq \widehat T
\}\sqrt{\frac{\lambda \sqrt{s}}{n \kappa_\cc} \cdot \(
\frac{(1+c)\lambda \sqrt{s}}{cn \kappa_\cc}  + 2 c_s\)}.
\end{array}$$
\end{theorem}

This theorem provides a performance bound for Post-LASSO as a function of LASSO's sparsity characterized by $\widehat m$, LASSO's rate of convergence, and LASSO's model selection ability.  For common designs this bound implies that Post-LASSO performs at least as well as LASSO, but it can be strictly better in some cases, and has a smaller shrinkage bias by construction.

\begin{corollary}[Asymptotic Bound for Post-LASSO]\label{corollary3:postrate} Suppose conditions of Corollary \ref{corollary1:rate} hold. Then
\begin{equation}
\|\widetilde \beta - \beta_0\|_{2,n} \lesssim_P \ \  \sigma \sqrt{ \frac{ s \log p}{n} } + c_s.
\end{equation}
If further $\widehat m = o(s)$ and $ T \subseteq \widehat T$ with probability approaching one, then

\begin{equation}\label{superior to LASSO}
\|\widetilde \beta - \beta_0\|_{2,n} \lesssim_P \ \ \sigma \[\sqrt{ \frac{ o(s)  \log p}{n} } +  \sqrt{ \frac{ s }{n} } \] + c_s.
 \end{equation}
If $\widehat T = T $ with probability approaching one, then Post-LASSO achieves the oracle performance
\begin{equation}\label{becomes oracle}
\|\widetilde \beta - \beta_0\|_{2,n} \lesssim_P \  \sigma \sqrt{ s/n } + c_s.
\end{equation}
\end{corollary}
It is also worth repeating here that finite-sample and asymptotic bounds in  other norms of interest immediately follow by the triangle inequality and by definition of $\kappa(\widehat m)$:
\begin{equation}\label{Def: Other Rates 3}
\sqrt{\En[ (x_i'\widetilde \beta - f_i)^2]} \leqslant  \|\widetilde \beta - \beta_0\|_{2,n}+c_s \ \ \text{and} \ \|\widetilde \beta - \beta_0\| \leqslant \|\widetilde \beta - \beta_0\|_{2,n}/ \kappa(\widehat m).
\end{equation}

The corollary above shows that Post-LASSO achieves the same near-oracle rate as LASSO.
Notably, this occurs despite the fact that LASSO may in general fail to correctly select the oracle model $T$ as a subset, that is $T \not \subseteq \widehat T$.  The intuition for this result is that any components of $T$ that LASSO misses cannot be very important. This corollary also shows that in some special cases Post-LASSO strictly improves upon LASSO's rate. Finally, note that Corollary \ref{corollary3:postrate} follows by observing that under the stated conditions,
\begin{equation}\label{Def: Other Rates 3}
\|\widetilde \beta - \beta_0\|_{2,n} \lesssim_P \ \  \sigma \[\sqrt{ \frac{ \widehat m \log p}{n} } +  \sqrt{ \frac{ s }{n} } + 1\{T \not\subseteq \widehat T  \} \sqrt{ \frac{s \log p}{ n}} \] + c_s.
\end{equation}

\section{Estimation of Noise Level}\label{Sec:DATAlambda}

Our specification of penalty levels (\ref{Def:Xdependent}) and (\ref{Def:Xindependent}) require the practitioner to know  the noise level $\sigma$ of the disturbances or at least estimate it. The purpose of this section is to propose the following method for estimating $\sigma$.  First, we
use a conservative estimate $\hat \sigma^{0} = \sqrt{\text{Var}_n[y_i]}:= \sqrt{\En\[(y_i - \bar y)^2\]}$, where  $\bar y = \En[y_{i}]$,
in place of $\sigma^2$ to obtain the initial LASSO and Post-LASSO estimates, $\widehat \beta$ and $\widetilde \beta$. The estimate
$\hat \sigma^{0}$ is conservative since $\hat \sigma^{0} = \sigma^{0} + o_P(1)$ where $\sigma^0 = \sqrt{\text{Var}[y_i]} \geqslant \sigma$, since $x_i$ contains
a constant by assumption.  Second, we define the refined estimate $\hat \sigma$ as
$$\hat \sigma= \sqrt{\widehat Q(\widehat \beta)}$$ in the case of LASSO and
$$\hat \sigma =  \sqrt{\frac{n}{n-\widehat s} \cdot \widehat Q(\widetilde \beta)} $$
in the case of Post-LASSO.
In the latter case we employ the standard degree-of-freedom correction
with $\widehat s = \|\widetilde \beta\|_0 = |\widehat T|$, and in the former case we need no additional
corrections, since the LASSO estimate is already sufficiently regularized.  Third, we use
the refined estimate $\hat \sigma^2$ to obtain the refined LASSO and Post-LASSO estimates $\widehat \beta$ and $\widetilde \beta$.   We can stop
here or further iterate on the last two steps.

Thus, the algorithm for estimating $\sigma$ using LASSO is as follows:
\begin{algorithm}[Estimation of $\sigma$ using LASSO iterations] Set $\hat \sigma^0 = \sqrt{\text{Var}_n[y_i]}$ and $k = 0$,  and specify a small constant $\nu >0$, the tolerance level, and
a constant $I>1$, the upper bound on the number of iterations.
(1)  Compute the LASSO estimator $\widehat \beta$ based on $\lambda = 2c \hat \sigma^k \Lambda(1-\alpha|X)$.
(2) Set $$\hat \sigma^{k+1} = \sqrt{\widehat Q(\widehat \beta)}.$$
(3) If $| \hat \sigma^{k+1} - \hat \sigma^{k}| \leqslant \nu$ or $k + 1 \geqslant I$, then stop and report $\hat\sigma = \hat \sigma^{k+1}$;
otherwise set $k \leftarrow k+1 $ and go to (1).
\end{algorithm}
And the algorithm for estimating $\sigma$ using Post-LASSO is as follows:
\begin{algorithm}[Estimation of $\sigma$ using Post-LASSO iterations] Set $\hat \sigma^0 = \sqrt{\text{Var}_n[y_i]}$ and $k = 0$,  and specify a small constant $\nu \geqslant 0$, the tolerance level, and
a constant $I>1$, the upper bound on the number of iterations.   (1)  Compute the Post-LASSO estimator $\widetilde \beta$ based on $\lambda = 2c \hat \sigma^k \Lambda(1-\alpha|X)$.
(2) Set $$\hat \sigma^{k+1} = \sqrt{\frac{n}{n-\widehat s} \cdot \widehat Q(\widetilde \beta)},$$ where $\widehat s = \|\widetilde \beta\|_0 = |\widehat T|$.
(3) If $| \hat \sigma^{k+1} - \hat \sigma^{k}| \leqslant \nu$ or $k + 1 \geqslant I$, then stop and report $\hat\sigma = \hat \sigma^{k+1}$; otherwise, set $k \leftarrow k+1 $ and go to (1).
\end{algorithm}
We can also use  $\lambda= 2c \hat \sigma^k \sqrt{n} \Phi^{-1}( 1- \alpha/2p) $ in place of $X$-dependent penalty. We note that using LASSO to estimate $\sigma$ it follows that the sequence $\widehat \sigma^k$, $k\geqslant 2$, is monotone, while using Post-LASSO the estimates $\widehat \sigma^k$, $k\geqslant 1$, can only assume a finite number of different values.

The following theorem shows that these algorithms produce consistent estimates of the noise level, and that
the  LASSO and Post-LASSO estimators based on the resulting data-driven penalty continue to obey the asymptotic bounds
we have derived previously.

\begin{theorem}[Validity of Results with Estimated $\sigma$] \label{theorem: sigma} Suppose conditions ASM and RES hold. Suppose that $\sigma \leqslant \hat \sigma^{0} \lesssim \sigma$ with probability approaching 1 and $s \log p/n \to 0$.  Then $\hat \sigma$ produced by either Algorithm 1 or 2 is consistent
 $$\hat \sigma/\sigma \to_P 1$$
 so that the penalty levels $\lambda= 2c \hat \sigma^k \Lambda(1-\alpha|X)$
and $\lambda= 2c \hat \sigma^k \sqrt{n} \Phi^{-1}( 1- \alpha/2p) $ with $\alpha = o(1)$, and $\log(1/\alpha) \lesssim \log p$, satisfy the condition (\ref{valid lambda}) of Corollary 1, namely
\begin{equation}\label{valid lambda2}
\lambda \lesssim_P \sigma\sqrt{n\log p} \ \text{ and } \  \lambda \geqslant c'n\|S\|_\infty  \text { wp  } \to 1,
\end{equation}
for some $1< c' < c$.  Consequently, the LASSO and Post-LASSO estimators based on this penalty level obey
the conclusions of Corollaries 1, 2, and 3.
\end{theorem}

\section{Monte Carlo Experiments}\label{Sec:MC}

In this section we compare the performance of LASSO,
Post-LASSO, and the ideal oracle linear regression estimators.
The oracle estimator applies ordinary least square to the true model. (Such an
estimator is not available outside Monte Carlo experiments.)

We begin by considering the following regression model:
$$ y = x'\beta_0 + \varepsilon, \  \ \beta_0 =(1,1,1/2,1/3,1/4,1/5,0,\ldots,0)', $$
where $x = (1,z')'$ consists of an intercept and covariates $z \sim N(0,\Sigma)$,
and the errors $\varepsilon$ are independently and identically
distributed $\varepsilon \sim N(0,\sigma^2)$. The dimension $p$ of the covariates $x$ is $500$, the dimension $s$ of the true model is $6$, and the sample size $n$ is $100$.  We set  $\lambda$ according to the $X$-dependent rule with $1-\alpha = 90\%$. The regressors are correlated with $\Sigma_{ij} = \rho^{|i-j|}$ and $\rho = 0.5$. We consider two levels of noise: Design 1 with $\sigma^2 = 1$ (higher level) and Design 2 with $\sigma^2=0.1$ (lower level). For each repetition we draw new vectors $x_i$'s and errors $\varepsilon_i$'s.

We summarize the model selection performance of LASSO in Figures
\ref{Fig:MCfirst01} and \ref{Fig:MCsecond01}.  In the left panels of the figures, we plot the frequencies of the dimensions of the selected model; in the right panels we plot the frequencies of selecting the correct regressors. From the left panels we see that the frequency of selecting a much larger model than the true model is very small in both designs. In the design with a larger noise, as the right panel of Figure \ref{Fig:MCfirst01} shows, LASSO frequently fails to select the entire true model, missing the regressors with small coefficients. However, it almost always includes the most important three regressors with the largest coefficients.  Notably, despite this partial failure of the model selection Post-LASSO still performs well, as we report below.  On the other hand, we see from the right panel of Figure \ref{Fig:MCsecond01} that in the design with a lower noise level LASSO rarely misses any component of the true support. These results confirm the theoretical results that when the non-zero coefficients are well-separated from zero, the penalized estimator should select a model that includes the true model as a subset.  Moreover, these results also confirm the theoretical result of Theorem \ref{Thm:Sparsity}, namely, that the dimension of the selected model should be of the same stochastic order as the dimension of the true model. In summary, the model selection performance of the penalized estimator agrees very well with the theoretical results.

We summarize the results on the performance of estimators in Table \ref{Table:MC}, which records for each estimator $\check \beta$  the mean $\ell_0$-norm $\Ep[ \|\check \beta\|_0 ]$, the norm of the bias $\|\Ep \check \beta - \beta_0\|$ and also the prediction error $\Ep[\En[ |x_i'( \check \beta - \beta_0 )|^2]^{1/2}]$ for recovering the regression function.  As expected, LASSO has a substantial bias. We see that Post-LASSO drastically improves upon the LASSO, particularly in terms of reducing the bias, which also results in a much lower overall prediction error.  Notably, despite that under the higher noise level LASSO frequently fails to recover the true model, the Post-LASSO estimator still performs well.  This is because the penalized estimator always manages to select the most important regressors.   We also see that the prediction error  of the Post-LASSO is within a factor $\sqrt{\log p}$ of the prediction error of the oracle estimator, as we would expect from our theoretical results.   Under the lower noise level, Post-LASSO performs almost identically to the ideal oracle estimator.  We would expect this since in this case LASSO selects the model especially well making Post-LASSO nearly the oracle.

\begin{figure}
\centering
\includegraphics[width=0.49\textwidth]{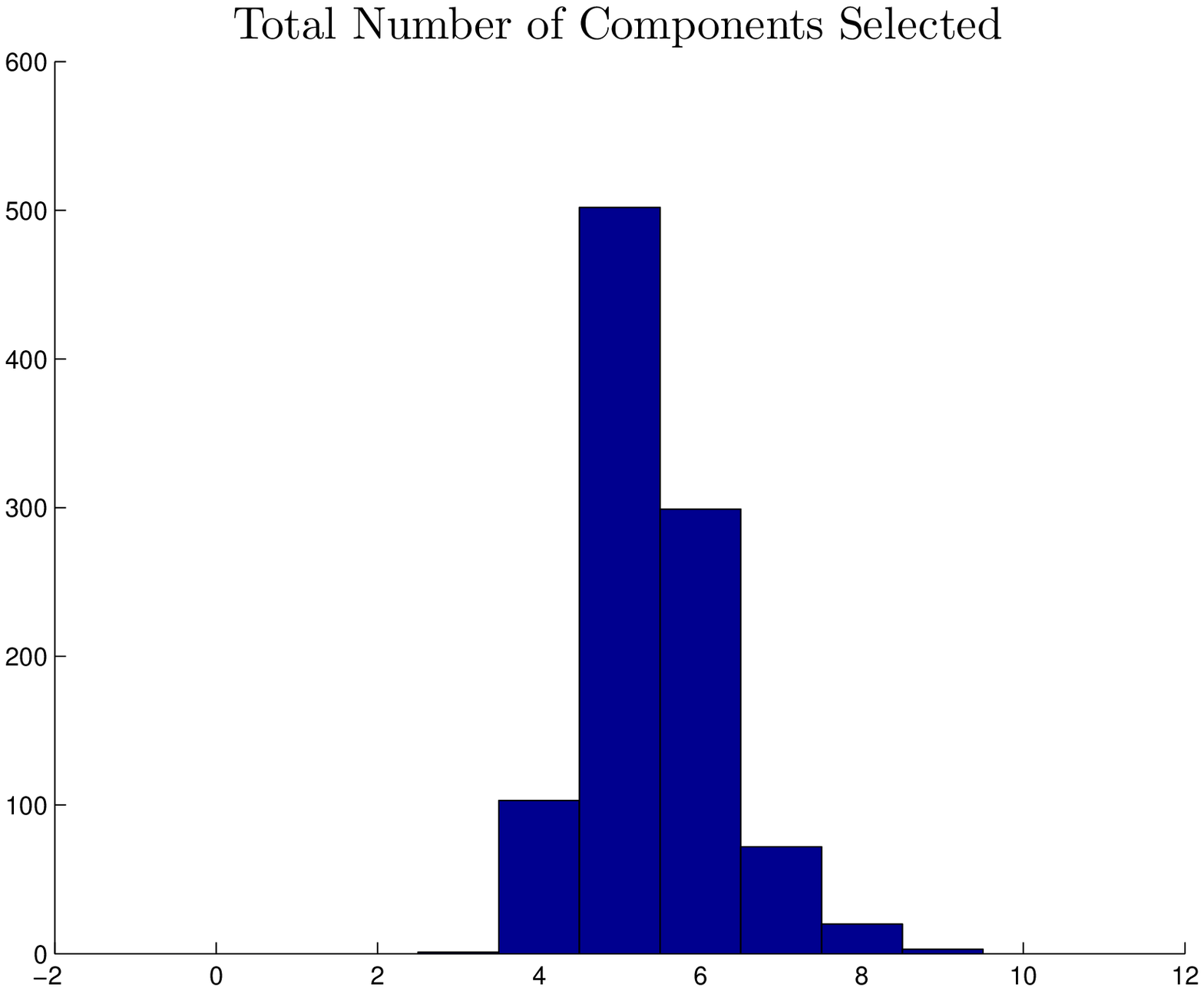}
\includegraphics[width=0.49\textwidth]{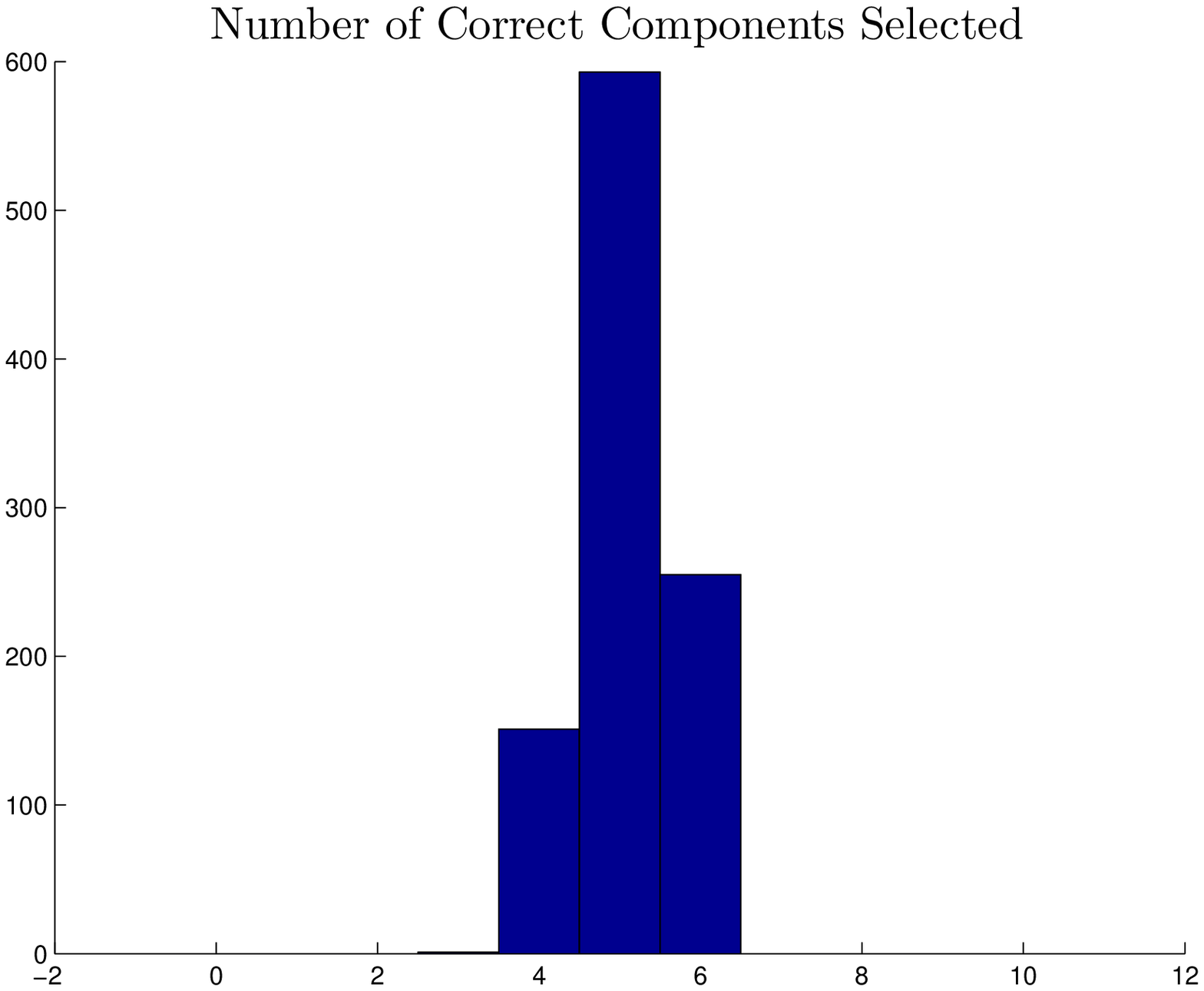}

\caption{ The figure summarizes the covariate selection results
for the design with $\sigma = 1$,  based on  $1000$ Monte
Carlo repetitions. The left panel plots the histogram for the
number of covariates selected by LASSO out of the possible 500 covariates, $|\widehat T|$.
The right panel plots the histogram for the number of significant
covariates selected by LASSO, $|\widehat T \cap T|$; there are in  total $6$ significant
covariates amongst $500$ covariates. The sample size for each repetition was $n=100$.}\label{Fig:MCfirst01}
\end{figure}

\begin{figure}
\centering
\includegraphics[width=0.49\textwidth]{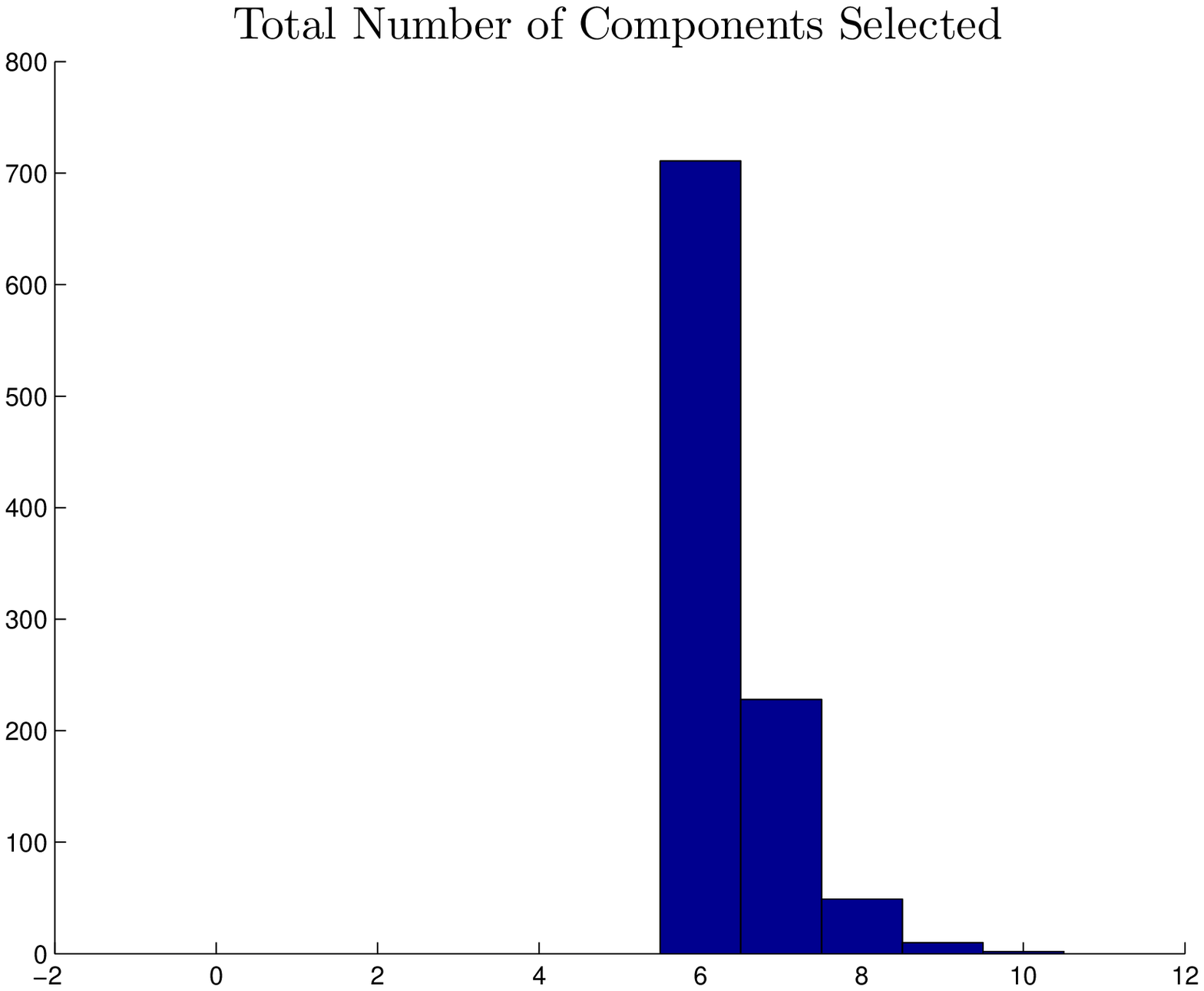}%
\includegraphics[width=0.49\textwidth]{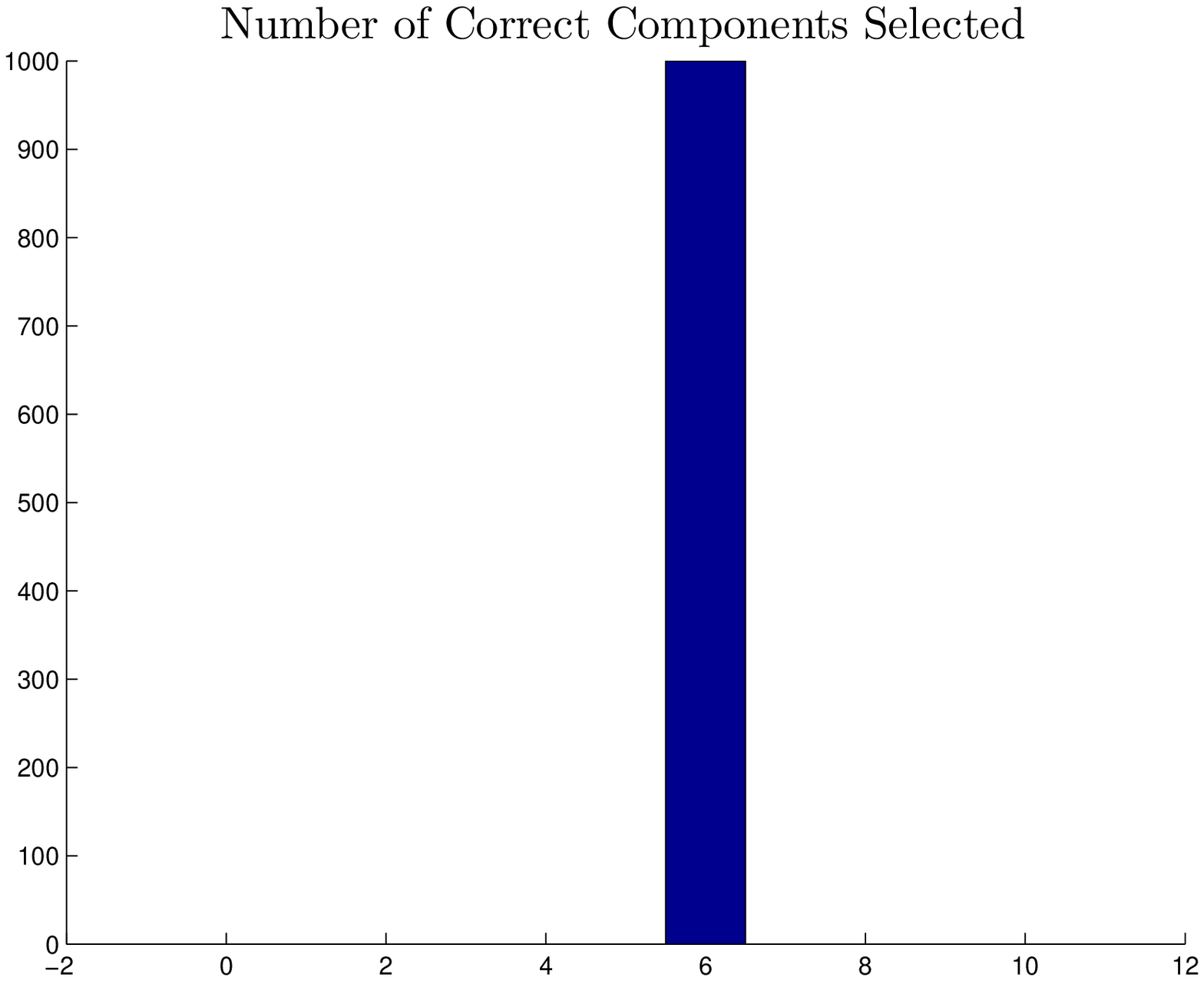}
\caption{The figure summarizes the covariate selection results for
the design with $\sigma^2 = 0.1$,  based on  $1000$ Monte Carlo repetitions.  The left
panel plots the histogram for the number of covariates selected
out of the possible 500 covariates, $|\widehat T|$. The right panel plots the
histogram for the number of significant covariates selected, $|\widehat T \cap T|$; there
are in total $6$ significant covariates amongst $500$ covariates. The sample size for each repetition was $n=100$.}\label{Fig:MCsecond01}
\end{figure}

\begin{table}
\begin{center}
{\bf Monte Carlo Results }

~\\

{\bf Design 1 ($\sigma^2 = 1$)}

\begin{tabular}{lccccc}
\hline
 & Mean $\ell_0$-norm &  Bias  & Prediction Error \\
  \hline
LASSO & 5.41 &  0.4136 & 0.6572 \\
  Post-LASSO & 5.41 &  0.0998 & 0.3298 \\
  Oracle  & 6.00 &  0.0122  & 0.2326 \\
\hline
\\
\end{tabular}

{\bf Design 2 ($\sigma^2 = 0.1$)}

\begin{tabular}{lccccc}
\hline
& Mean $\ell_0$-norm &  Bias  & Prediction Error \\
  \hline
LASSO & 6.3640 &  0.1395 & 0.2183 \\
  Post-LASSO & 6.3640   & 0.0068 & 0.0893 \\
  Oracle  & 6.00 &  0.0039  & 0.0736 \\
\hline
\\
\end{tabular}
\end{center}
\caption{The table displays the average $\ell_0$-norm
of the estimators as well as mean bias and prediction error. We
obtained the results using $ 1000 $ Monte Carlo repetitions for
each design.} \label{Table:MC}
\end{table}

The results above used the true value of $\sigma$ in the choice of $\lambda$. Next we illustrate how $\sigma$ can be estimated in practice. We follow the iterative procedure described in the previous section. In our experiments the tolerance was $10^{-8}$ times the current estimate for $\sigma$, which is typically achieved in less than 15 iterations.

We assess the performance of the iterative procedure under the design with the larger noise, $\sigma^2 = 1$ (similar results hold for $\sigma^2 = 0.1$).
The histograms in Figure \ref{Fig:MCestimated01} show that the model selection properties are very similar to the model selection when $\sigma$ is known.
Figure \ref{Fig:MCestimateed02} displays the distribution of the estimator $\hat \sigma$ of $\sigma$ based on (iterative) Post-LASSO, (iterative) LASSO, and the initial estimator $\hat\sigma^0=\sqrt{{\rm Var}_n[y_i]}$. As we expected, estimator $\hat \sigma$ based on LASSO
produces estimates that are somewhat higher than the true value.  In contrast, the estimator $\hat \sigma$ based on Post-LASSO seems to perform very well in our experiments, giving estimates
$\hat \sigma$ that bunch closely near the true value $\sigma$.

\begin{figure}[!p]
\centering
\includegraphics[width=0.49\textwidth]{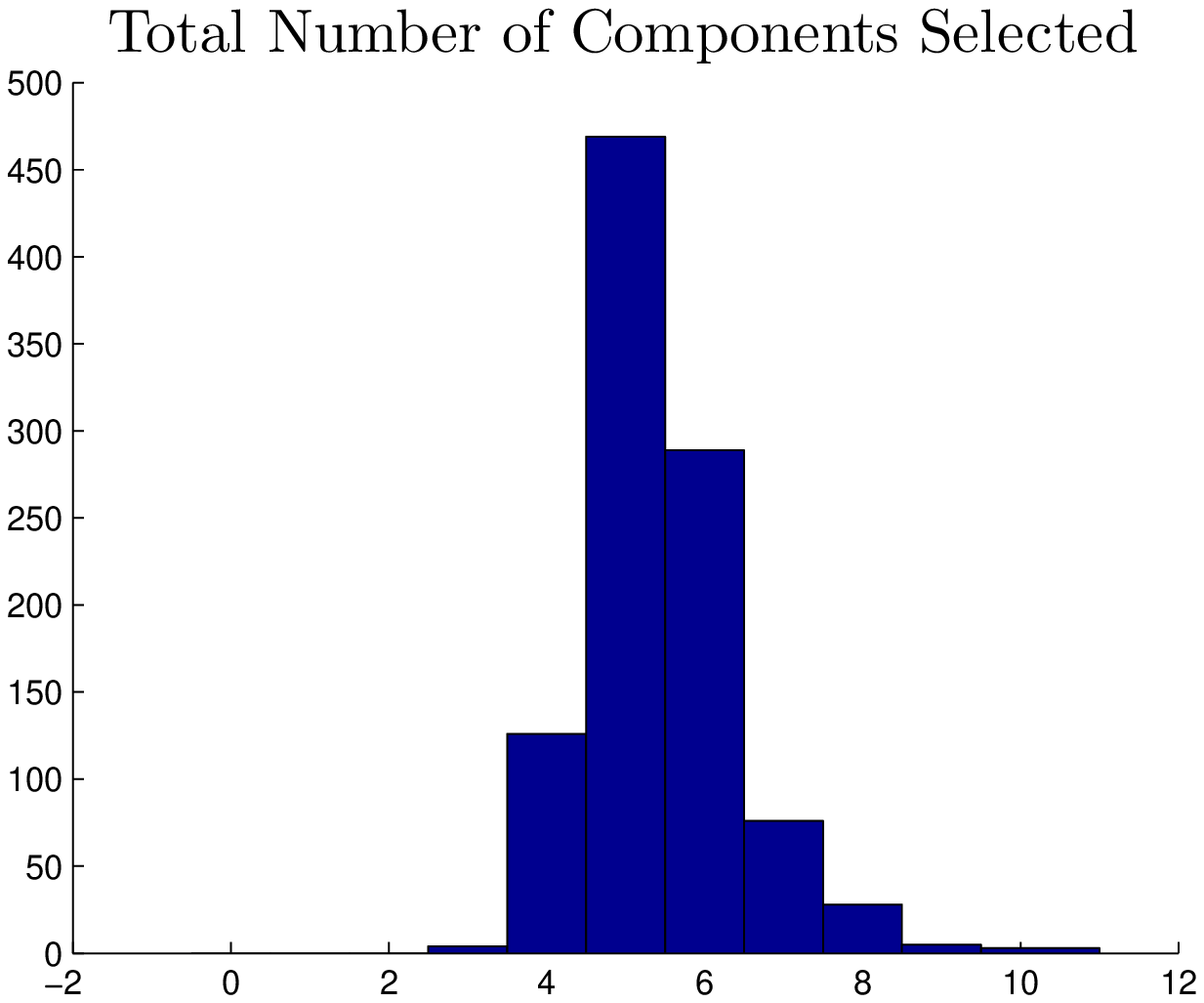}
\includegraphics[width=0.49\textwidth]{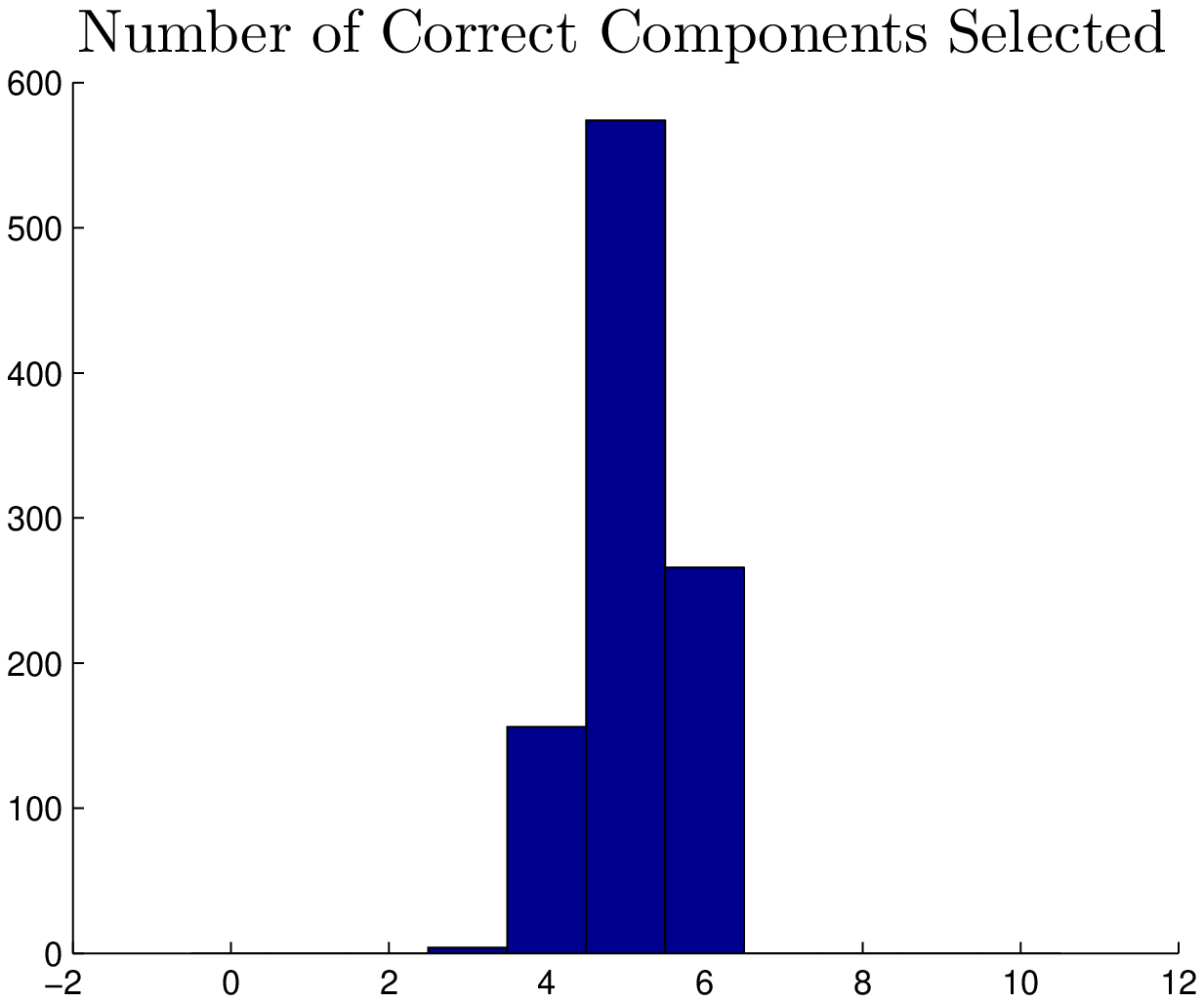}
\caption{The figure summarizes the covariate selection results
for the design with $\sigma = 1$, when $\sigma$ is estimated, based on  $1000$ Monte
Carlo repetitions. The left panel plots the histogram for the
number of covariates selected out of the possible 500 covariates.
The right panel plots the histogram for the number of significant
covariates selected; there are in  total $6$ significant
covariates amongst $500$ covariates. The sample size for each repetition was $n=100$.}\label{Fig:MCestimated01}
\end{figure}

\begin{figure}[!p]
\centering
\includegraphics[width=0.49\textwidth]{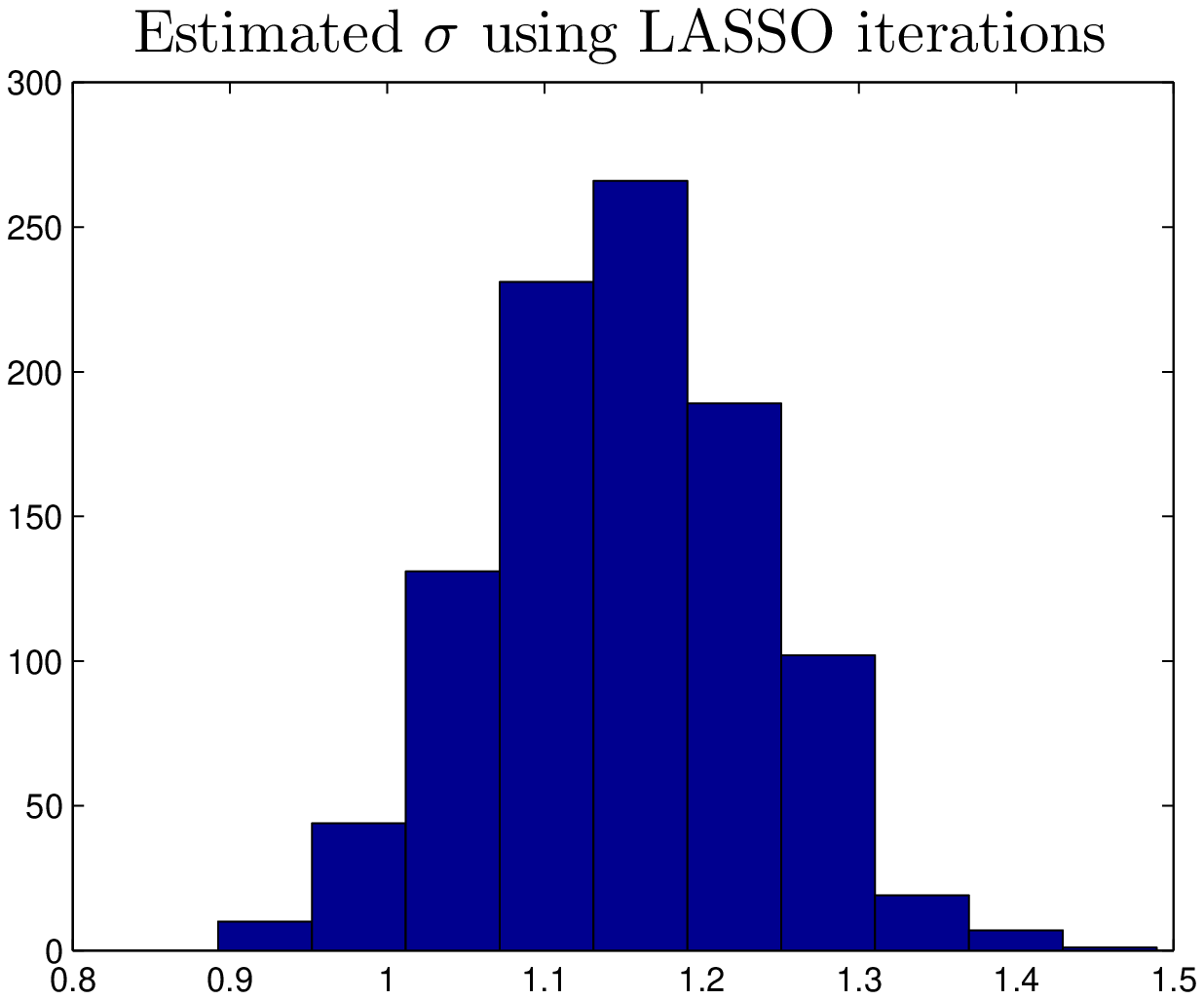}
\includegraphics[width=0.49\textwidth]{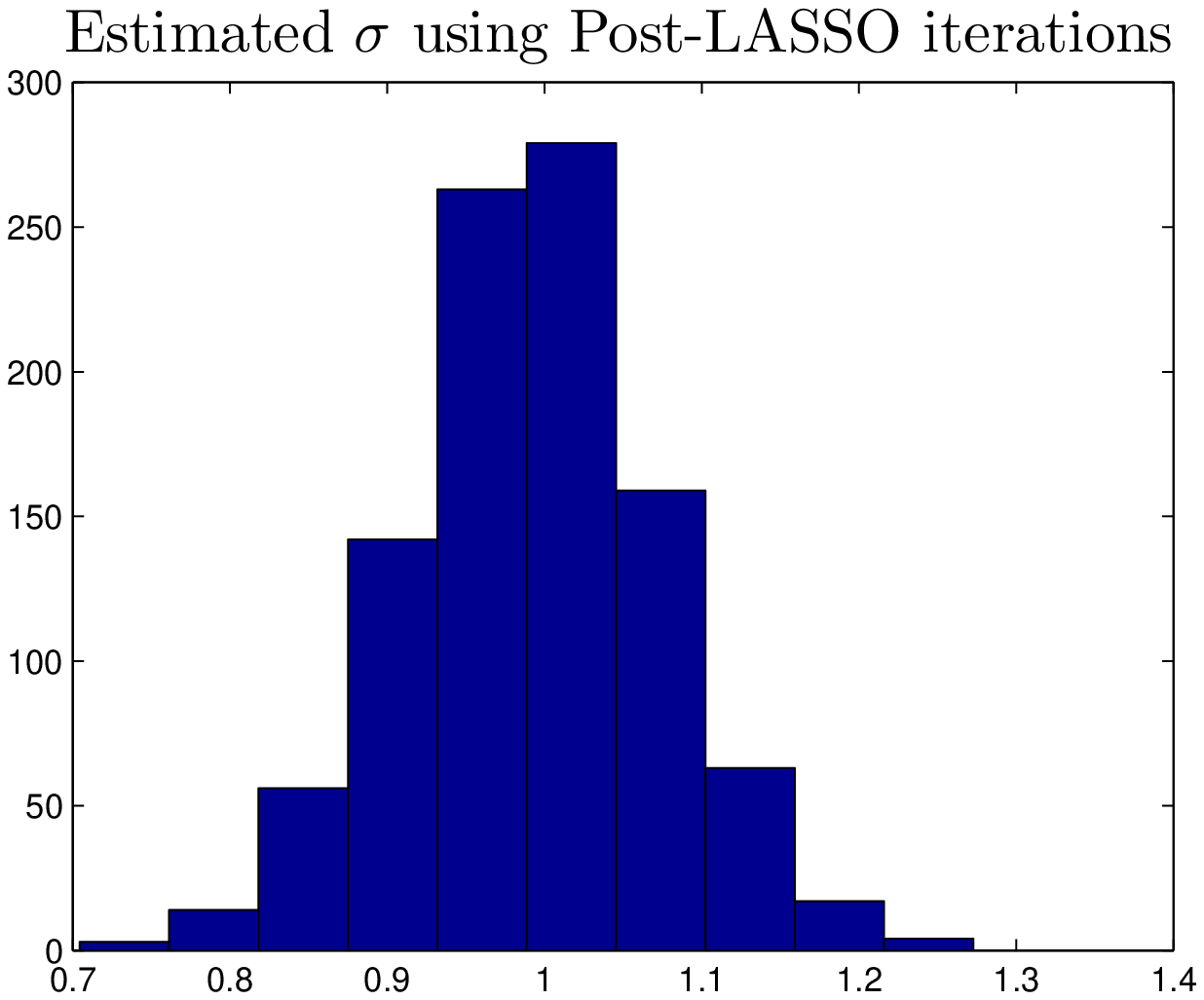}
\includegraphics[width=0.49\textwidth]{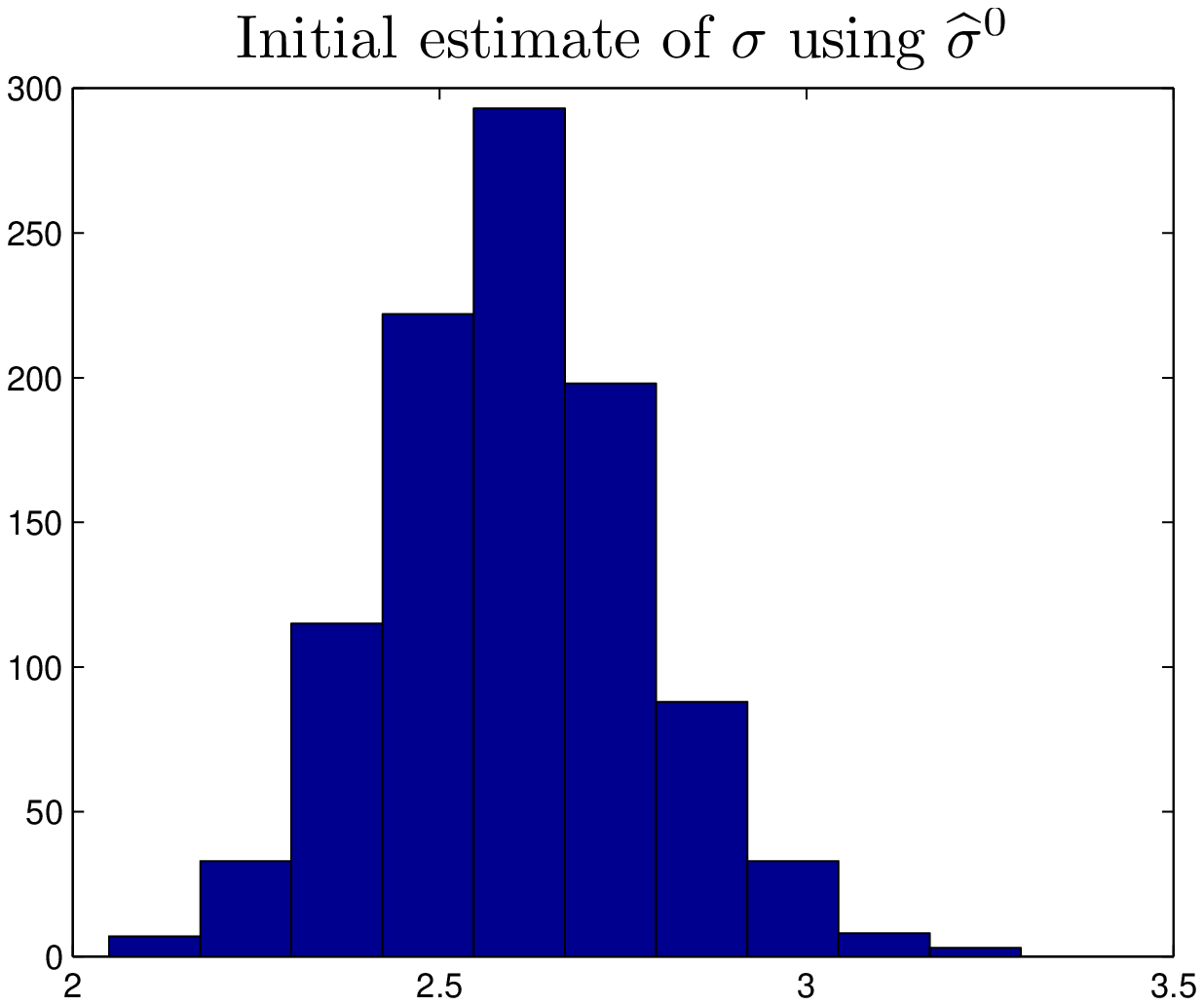}
\caption{The figure displays the distribution of the estimator $\widehat \sigma$ of $\sigma$ based on (iterative) LASSO, (iterative) Post-LASSO, and the conservative initial estimator $\widehat \sigma^0 = \sqrt{{\rm Var}_n[y_i]}$. The plots summarize the estimation performance for the design with $\sigma = 1$, based on  $1000$ Monte Carlo repetitions.}\label{Fig:MCestimateed02}
\end{figure}

\section{Application to Cross-Country Growth Regression}\label{Sec:Growth}

In this section we apply LASSO and Post-LASSO to an international economic growth example. We use the Barro and Lee \cite{BarroLee1994} data consisting of a panel of 138 countries for the period of 1960 to 1985. We consider the national growth rates in GDP per capita as a dependent variable $y$ for the periods 1965-75 and 1975-85.\footnote{The growth rate in GDP over a period from $t_1$ to $t_2$ is commonly defined as $\log (GDP_{t_2}/GDP_{t_1})$.}  In our analysis, we will consider a model with $p=62$ covariates, which allows for a total of $n=90$ complete observations.  Our goal here is to select a subset of these covariates and briefly compare the resulting models to the standard models used in the empirical growth literature (Barro and Sala-i-Martin \cite{BarroSala1995}).

Let us now turn to our empirical results. We performed covariate selection using LASSO, where we used our data-driven choice of penalty level
$\lambda$ in two ways. First we used an upper bound on $\sigma$ being $\widehat \sigma^0$ and decreased the penalty to estimate different models with
$\lambda$, $\lambda/2$, $\lambda/3$, $\lambda/4$, and $\lambda/5$. Second, we applied the iterative procedure described in the previous section to
define $\lambda^{it}$ (which is computed based on $\hat\sigma^{it}$ obtained using the iterative Post-LASSO procedure).

The initial choice of the first approach led us to select no covariates, which is consistent with over-regularization since an upper bound for $\sigma$ was used. We then proceeded to slowly decrease the penalty level in order to allow for some covariates to be selected.   We present the model selection results in Table \ref{Table:Growth}. With the first
relaxation of the choice of $\lambda$, we select the black market exchange rate premium (characterizing trade openness) and a measure of political
instability. With a second relaxation of the choice of $\lambda$ we select an additional set of variables reported in the table. The iterative approach led to a model with only the black market exchange  premium. We refer the reader to \cite{BarroLee1994} and \cite{BarroSala1995} for a complete definition and discussion of each of these variables.

We then proceeded to apply ordinary linear regression to the selected models and we also report the standard confidence intervals for these estimates. Table \ref{Table:CondidenceInterval} shows these results.
We find that in all models with additional selected covariates, the linear regression coefficients on the initial level of GDP is always negative and
the standard confidence intervals do not include zero. We believe that these empirical findings firmly support the hypothesis of (conditional) convergence derived from
the classical Solow-Swan-Ramsey growth model.\footnote{The inferential method used here is actually valid under certain conditions, despite the fact that the model has been selected; this is demonstrated in a work in progress.}
Finally, our findings also agree with and thus support the previous findings reported in  Barro and Sala-i-Martin \cite{BarroSala1995}, which relied on ad-hoc reasoning for covariate selection.

\begin{table}
\begin{center}
{\bf Confidence Intervals after Model Selection \\ for the International Growth Regressions}

\begin{tabular}{ccccc}
\\
\hline Penalization & & \multicolumn{2}{c}{Real GDP per capita (log)} \\ Parameter & & \\ $\lambda=2.7870$ & & Coefficient & $90\%$ Confidence Interval \\
\hline
$\lambda^{it}=2.3662$ & & $-0.0112$      & $[-0.0219,  -0.0007]$\\
$\lambda/2$ & & $-0.0120$      & $[ -0.0225,   -0.0015]$\\
$\lambda/3$ & & $-0.0153$      & $[-0.0261,   -0.0045]$\\
$\lambda/4$ & & $-0.0221$      & $[-0.0346,   -0.0097]$\\
$\lambda/5$ & & $-0.0370$      & $[-0.0556,   -0.0184]$\\
\hline

\end{tabular}\caption{The table above displays the coefficient and a $90\%$ confidence interval associated with each model selected by the
corresponding penalty level. The selected models are displayed in Table \ref{Table:Growth}.}\end{center}\label{Table:CondidenceInterval}
\end{table}


{\small
\begin{table}\begin{center}{ \bf Model Selection Results for the International Growth Regressions}
\renewcommand{\arraystretch}{1}

 \begin{tabular}{ccc}
\hline Penalization & & \\ Parameter & & Real GDP per capita (log) is included in all models \\ 
$\lambda = 2.7870$ &  & Additional Selected Variables \\ \hline \rowcolor[gray]{0.9} $\lambda$ &  & -  \\ $\lambda^{it}$ &  & Black Market Premium (log)  \\ \rowcolor[gray]{0.9} $\lambda/2$  &  &  Black Market Premium (log)  \\ \rowcolor[gray]{0.9}                 &  &  Political Instability \\ $\lambda/3$  &  &  Black Market Premium (log)  \\
                  &  &  Political Instability \\
                  &  & Ratio of nominal government expenditure on defense to nominal GDP\\ 
                  &  & Ratio of import to GDP\\ 
\rowcolor[gray]{0.9}  $\lambda/4$ &  & Black Market Premium (log)  \\ \rowcolor[gray]{0.9}                  &  & Political Instability \\ \rowcolor[gray]{0.9}                  &  & Ratio of nominal government expenditure on defense to nominal GDP\\ 

$\lambda/5$                      &  & Black Market Premium (log)  \\ 
                  &  & Political Instability \\ 
                  &  & Ratio of nominal government expenditure on defense to nominal GDP\\ 
                  &  & Ratio of import to GDP\\ 
                  &  & Exchange rate \\ 
                  &  &  \% of ``secondary school complete" in male population \\ 
                  &  & Terms of trade shock \\ 
                  &  & Measure of tariff restriction  \\ 
                  &  & Infant mortality rate \\ 
                  &  & Ratio of real government ``consumption" net of defense and education\\ 
                  &  & Female gross enrollment ratio for higher education  \\ 
\hline
\end{tabular}\caption{The models selected at various levels of penalty.}\label{Table:Growth}
\end{center}
\end{table}}

\begin{acknowledgement}
We would like to thank Denis Chetverikov and Brigham Fradsen
for thorough proof-reading of several versions of this paper and their
detailed comments that helped us considerably improve the paper.
We also would like to thank  Eric Gautier, Alexandre Tsybakov,  and two anonymous referees for their comments that also helped us considerably improve the chapter.
We would also like to thank the participants of seminars in Cowles Foundation Lecture at the Econometric Society Summer
Meeting, Duke University, Harvard-MIT, and the Stats in the Chateau.
\end{acknowledgement}

\section*{Appendix}

\addcontentsline{toc}{section}{Appendix}

\section{Proofs}

\begin{proof}[Theorem \ref{Thm:Nonparametric}] Proceeding similarly to \cite{BickelRitovTsybakov2009}, by optimality of $\widehat \beta$ we have that
\begin{equation}\label{Def:Opt} \widehat Q(\widehat\beta) - \widehat Q(\beta_0) \leqslant \frac{\lambda}{n}\|\beta_0\|_1 -
\frac{\lambda}{n}\|\widehat\beta\|_1.\end{equation}
To prove the result we make the use of the following relations: for $\delta = \hat \beta - \beta_0$, if $\lambda \geqslant cn\|S\|_{\infty}$
 \begin{eqnarray}\label{Rel:(a)2}
 \hat Q (\hat \beta) - \hat Q(\beta_0) -  \| \delta\|^2_{2,n} & = & - 2 \En [\varepsilon_i x_i'\delta]  - 2 \En [r_i x_i' \delta]\\
& \geqslant &  - \|S\|_{\infty} \| \delta\|_{1}  -  2 c_s \| \delta\|_{2,n} \nonumber \\
 & \geqslant &  - \frac{\lambda}{cn} (\| \delta_T\|_{1} + \| \delta_{T^c}\|_{1})  -  2c_s \| \delta\|_{2,n},
 \end{eqnarray}
 \begin{eqnarray}\label{Rel:(b)2}
 \|\beta_0\|_{1} - \| \hat \beta\|_{1} &= &  \| \beta_{0T}\|_{1} - \|\hat \beta_T\|_{1} - \| \hat \beta_{T^c} \|_{1}
  \leqslant   \| \delta_T\|_{1} - \| \delta_{T^c}\|_{1}.
 \end{eqnarray}
Thus, combining (\ref{Def:Opt}) with (\ref{Rel:(a)2})--(\ref{Rel:(b)2}) implies that
\begin{equation}\label{keykey}
-\frac{\lambda}{cn} (\| \delta_T\|_{1} + \| \delta_{T^c}\|_{1}  ) + \|\delta\|_{2,n}^2 - 2 c_s \| \delta\|_{2,n} \leqslant \frac{\lambda}{n} (
\|\delta_T\|_{1} - \|\delta_{T^c}\|_{1}).
\end{equation}
If $\|\delta\|_{2,n}^2 - 2 c_s \| \delta\|_{2,n}<0$, then we have established the bound in the statement of the theorem. On the other hand, if
$\|\delta\|_{2,n}^2 - 2 c_s \| \delta\|_{2,n} \geqslant 0$ we get
 \begin{equation}\label{domination3}
\|\delta_{T^c}\|_{1} \leqslant \frac{c+1}{c-1}  \cdot \| \delta_{T}\|_{1} = \cc \| \delta_{T}\|_{1},
 \end{equation}
and therefore $\delta$ satisfies the condition to invoke RE($c$). From (\ref{keykey}) and using RE($c$), $\|\delta_T\|_1 \leqslant \sqrt{s}\|\delta\|_{2,n}/\kappa_\cc$, we get
$$
\|\delta\|_{2,n}^2 - 2 c_s \| \delta\|_{2,n} \leqslant  \(1 + \frac{1}{c}\) \frac{\lambda}{n}\|\delta_T\|_{1} \leqslant  \(1 + \frac{1}{c}\)\frac{\lambda\sqrt{s}}{n}\frac{\|\delta\|_{2,n}}{\kappa_\cc} $$ which gives the result on the prediction norm. 
\end{proof}


\begin{lemma}[Empirical pre-sparsity for LASSO]\label{Lemma:SparsityLASSO}
In either the parametric model or the nonparametric model, let $\hat m = |\widehat T \setminus T|$ and $\lambda \geqslant c \cdot n\|S\|_\infty$. We have
$$ \sqrt{\hat m} \leqslant \sqrt{s}\sqrt{\phi(\hat m)} \ 2\cc/\kappa_\cc + 3(\cc+1) \sqrt{\phi(\hat m)} \ nc_s/\lambda,$$ where $c_s = 0$ in the parametric model.
\end{lemma}
\begin{proof}  We have from the optimality conditions that  $$2\En[ x_{ij}(y_i-x_i'\hat\beta)] = \sign(\hat\beta_j)\lambda/n \ \text{ for each } \ j \in \widehat T\setminus T.  $$

Therefore we have for $R=(r_1,\ldots,r_n)'$, $X = [x_1,...,x_n]'$, and $Y = (y_1,...,y_n)'$
$$\begin{array}{rcl}
  \sqrt{\hat m}\lambda   & = &  2\| (X'(Y - X\hat \beta))_{\widehat T\setminus T} \| \\
            & \leqslant &  2\| (X'(Y - R - X \beta_0))_{\widehat T\setminus T} \| + 2\| (X'R)_{\widehat T\setminus T} \| + 2\| (X'X(\beta_0-\hat \beta))_{\widehat T\setminus T} \| \\
            & \leqslant & \sqrt{\hat m}\cdot n\|S\|_{\infty} + 2 n\sqrt{\phi(\hat m)} c_s+  2n \sqrt{ \phi(\hat m )} \|\hat \beta - \beta_0 \|_{2,n},\\
\end{array}
$$
where we used that
$$\begin{array}{rcl}
\| (X'X(\beta_0-\hat \beta))_{\widehat T\setminus T} \|&
\leqslant &\sup_{\|v_{T^c}\|_0\leqslant \hat m, \|v\|\leqslant 1}|
v' X'X(\beta_0-\hat \beta)| \\
&\leqslant &
\sup_{\|v_{T^c}\|_0\leqslant \hat m, \|v\|\leqslant 1}\| v'X'\|\|X(\beta_0-\hat \beta)\| \\
&= & \sup_{\|v_{T^c}\|_0\leqslant \hat
m, \|v\|\leqslant 1}\sqrt{| v'X'Xv|}\|X(\beta_0-\hat \beta)\|\\
& = &
n\sqrt{\phi(\hat m)}\|\beta_0-\hat \beta\|_{2,n},\\
\end{array}$$ and
similarly $\| (X'R)_{\widehat T\setminus T} \| \leqslant
n\sqrt{\phi(\hat m)} c_s$.

Since $\lambda/c \geqslant n\|S\|_\infty$, and by Theorem \ref{Thm:Nonparametric}, $\|\beta_0-\hat\beta\|_{2,n} \leqslant \(1 + \frac{1}{c}\) \frac{\lambda \sqrt{s}}{n \kappa_\cc} + 2c_s$, we have
$$ (1-1/c)\sqrt{\hat m} \leqslant 2\sqrt{\phi(\hat m)}(1+1/c)\sqrt{s}/\kappa_\cc + 6 \sqrt{\phi(\hat m)}\ nc_s/\lambda.$$

The result follows by noting that $(1-1/c) = 2/(\cc+1)$ by definition of $\cc$.
\end{proof}

\begin{proof}[Proof of Theorem \ref{Thm:Sparsity}]
Since  $\lambda \geqslant c\cdot n\|S\|_\infty$ by Lemma \ref{Lemma:SparsityLASSO} we have
$$ \sqrt{\hat m} \leqslant \sqrt{ \phi(\hat m)} \cdot 2\cc\sqrt{s}/\kappa_\cc + 3(\cc+1) \sqrt{\phi(\hat m)} \cdot nc_s/\lambda,$$
which, by letting $L = \left( \frac{2\cc}{\kappa_\cc} + 3(\cc+1)\frac{nc_s}{\lambda\sqrt{s}}\right)^2$, can be rewritten as
 \begin{equation}\label{Eq:Sparsity}\hat m \leqslant s \cdot \phi(\hat m) L.\end{equation}
Note that $\widehat m \leqslant n$ by optimality conditions.  Consider any $M \in \mathcal{M}$, and suppose $\widehat m > M$. Therefore by Lemma \ref{Lemma:SparseEigenvalueIMP} on sublinearity of sparse eigenvalues
$$ \hat m \leqslant s \cdot \ceil{\frac{\hat m}{M}}\phi(M) L.$$ Thus, since $\ceil{k}< 2k$ for any $k\geqslant 1$  we have
$$ M <  s \cdot 2\phi(M) L$$
which violates the condition of $M \in \mathcal{M}$ and $s$. Therefore, we must have $\widehat m \leqslant M$.

In turn, applying (\ref{Eq:Sparsity}) once more with $\widehat m \leqslant (M\wedge n)$ we obtain
 $$ \hat m \leqslant s \cdot \phi(M\wedge n) L.$$
The result follows by minimizing the bound over $M \in \mathcal{M}$.
\end{proof}


\begin{proof}[Lemma \ref{Lemma:Crack}, part (1)]
The result follows immediately from the assumptions.
\end{proof}

\begin{proof}[Lemma \ref{Lemma:Crack}, part (2)]
Let $\widehat m = |\widehat T\setminus T|=\|\widehat \beta_{T^c}\|_0$. Then, note that $\|\delta\|_\infty \leqslant \|\delta\| \leqslant \|\delta\|_{2,n}/\kappa(\widehat m)$. The result
follows from Theorem \ref{Thm:Nonparametric}.
\end{proof}

\begin{proof}[Lemma \ref{Lemma:Crack}, part (3)]
Let $\delta:=\widehat\beta-\beta_0$. Note that by the first order optimality conditions of $\widehat \beta$ and the assumption on $\lambda$
$$\begin{array}{rl}\displaystyle \|\En[x_ix_i'\delta]\|_{\infty} &\displaystyle \leqslant \|\En[x_i(y_i-x_i'\widehat\beta)]\|_{\infty} + \|S/2\|_{\infty} + \|\En[x_ir_i]\|_{\infty} \\
&\displaystyle \leqslant \frac{\lambda}{2n} +
\frac{\lambda}{2cn} + \min\left\{\frac{\sigma}{\sqrt{n}}, c_s\right\}\end{array}$$ since $\|\En[x_ir_i]\|_{\infty} \leqslant \min\left\{\frac{\sigma}{\sqrt{n}}, c_s\right\}$ by Lemma \ref{lemma:xiri} below.

Next let $e_j$ denote the $j$th-canonical direction. Thus, for every $j=1,\ldots,p$ we have
$$\begin{array}{rl}
\displaystyle  | \En[e_j'x_ix_i'\delta] - \delta_j | =
 |\En[e_j'(x_ix_i'-I)\delta]| &\displaystyle  \leqslant  \max_{1\leqslant j,k \leqslant p}|(\En[x_ix_i'-I])_{jk}| \ \|\delta\|_{1}\\
&\displaystyle   \leqslant  \|\delta\|_1/[Us].\end{array} $$
Then, combining the two bounds above and using the triangle inequality we have
$$ \|\delta\|_\infty \leqslant  \|\En[x_ix_i'\delta]\|_{\infty} +  \|\En[x_ix_i'\delta]-\delta\|_{\infty} \leqslant \(1+\frac{1}{c}\)\frac{\lambda}{2n} + \min\left\{\frac{\sigma}{\sqrt{n}}, c_s\right\} +
\frac{\|\delta\|_1}{Us}.$$
The result follows by Lemma \ref{Lemma:L1} to bound $\|\delta\|_1$ and the
arguments in \cite{BickelRitovTsybakov2009} and \cite{Lounici2008} to show that the bound on the correlations imply that for any $C>0$
$$ \kappa_C \geqslant \sqrt{ 1 - s(1+2C)\|\En[x_ix_i'-I]\|_\infty}$$ so that
$\kappa_\cc \geqslant \sqrt{1-[(1+2\cc)/U]}$ and $\kappa_{2\cc} \geqslant \sqrt{1-[(1+4\cc)/U]}$ under this particular design. \end{proof}

\begin{lemma}\label{lemma:xiri}
Under condition ASM, we have that
$$ \| \En[ x_ir_i ]\|_\infty \leqslant \min\left\{\frac{\sigma}{\sqrt{n}}, c_s\right\}.$$
\end{lemma}
\begin{proof}
First note that for every $j=1,\ldots, p$, we have $|\En[x_{ij}r_i]|\leqslant \sqrt{\En[x_{ij}^2]\En[r_i^2]}=c_s$.

Next, by definition of $\beta_0$ in (\ref{oracle}), for $j \in T$ we have
$$\En[x_{ij}(f_i-x_i'\beta_0)] = \En[x_{ij}r_i] = 0 $$ since $\beta_0$ is a minimizer over the support of $\beta_0$. For $j \in T^c$ we have that for any $t \in \RR$
 $$ \En [(f_i - x_i'\beta_0)^2] + \sigma^2 \frac{s}{n} \leqslant \En [(f_i - x_i'\beta_0-tx_{ij})^2] + \sigma^2 \frac{s+1}{n}.
$$ Therefore, for any $t\in \RR$ we have
$$ -\sigma^2/n \leqslant  \En [(f_i - x_i'\beta_0-tx_{ij})^2]  - \En [(f_i - x_i'\beta_0)^2] = -2t\En [x_{ij}(f_i - x_i'\beta_0)]+t^2\En[x_{ij}^2].$$
Taking the minimum over $t$ in the right hand side at $t^* = \En [x_{ij}(f_i - x_i'\beta_0)]$ we obtain
$$ -\sigma^2/n \leqslant - (\En [x_{ij}(f_i - x_i'\beta_0)])^2$$ or equivalently, $|\En [x_{ij}(f_i - x_i'\beta_0)]|\leqslant \sigma/\sqrt{n}$.

\end{proof}

\begin{lemma}\label{Lemma:L1}
If $\lambda \geqslant c n\| S\|_{\infty}  $, then for $\cc = (c+1)/(c-1)$ we have
\begin{eqnarray*}
&& \|\widehat \beta - \beta_0\|_{1}  \leqslant  \frac{(1+2\cc)\sqrt{s}}{\kappa_{2\cc}} \left[  \(1 + \frac{1}{c}\) \frac{\lambda \sqrt{s}}{n \kappa_\cc} + 2 c_s\right] + \(1 + \frac{1}{2\cc}\)\frac{2c}{c-1}\frac{n}{\lambda} c_s^2,
 \end{eqnarray*}
 where $c_s = 0$ in the parametric case.
\end{lemma}
\begin{proof}
First, assume $\|
\delta_{T^c}\|_1 \leqslant 2\cc \| \delta_T\|_1.$ In this
case, by definition of the restricted eigenvalue, we have
$$ \|\delta\|_1 \leqslant (1+2\cc) \|\delta_T\|_1 \leqslant (1+2\cc)\sqrt{s}\|\delta\|_{2,n}/\kappa_{2\cc} $$
and the result follows by applying the first bound to
$\|\delta\|_{2,n}$  since $\cc > 1$.

On the other hand, consider the case that $\|\delta_{T^c}\|_1 > 2\cc \|\delta_T\|_1$ which would already imply $\|\delta\|_{2,n} \leqslant 2c_s$. Moreover, the relation (\ref{keykey}) implies that
$$\begin{array}{rcl}
 \|\delta_{T^c}\|_1 & \leqslant & \cc \|\delta_T\|_1 + \frac{c}{c-1} \frac{n}{\lambda}\|\delta\|_{2,n}(2c_s - \|\delta\|_{2,n})\\
 & \leqslant & \cc \|\delta_T\|_1 + \frac{c}{c-1}\frac{n}{\lambda} c_s^2\\
& \leqslant & \frac{1}{2} \| \delta_{T^c}\|_1 + \frac{c}{c-1}\frac{n}{\lambda} c_s^2.\\
\end{array}
$$

 Thus,
$$\| \delta\|_1 \leqslant \(1 + \frac{1}{2\cc}\) \| \delta_{T^c}\|_1 \leqslant \(1 + \frac{1}{2\cc}\)\frac{2c}{c-1}\frac{n}{\lambda} c_s^2. $$

The result follows by adding the bounds on each case and invoking Theorem \ref{Thm:Nonparametric} to bound $\|\delta\|_{2,n}$.
\end{proof}


\begin{proof}[Theorem \ref{Cor:2StepNonparametric}]
Let $\widetilde \delta := \widetilde \beta - \beta_0$. By definition of the Post-LASSO estimator, it follows that $\widehat Q(\widetilde \beta) \leqslant \widehat Q(\widehat \beta)$ and $\widehat Q(\widetilde \beta) \leqslant \widehat Q(\beta_{0\widehat T})$. Thus,
$$
\widehat Q (\widetilde \beta) - \widehat Q(\beta_0) \leqslant  \(  \widehat Q (\widehat \beta) - \widehat Q(\beta_0)  \)  \wedge  \( \widehat
Q(\beta_{0\widehat T}) - \widehat Q(\beta_0) \)  =: B_n \wedge C_n.$$


The least squares criterion function satisfies
$$\begin{array}{rcl} | \widehat Q(\widetilde \beta) - \widehat Q(\beta_0) - \|\widetilde \delta\|^2_{2,n} | & \leqslant & |S'\widetilde \delta| +2c_s\|\widetilde \delta\|_{2,n} \\
& \leqslant & |S_T'\widetilde\delta| + |S_{T^c}'\widetilde\delta| + 2c_s\|\widetilde \delta\|_{2,n} \\
& \leqslant & \|S_T\|\|\widetilde \delta\| + \|S_{T^c}\|_{\infty} \|\widetilde \delta_{T^c} \|_1 + 2c_s\|\widetilde \delta\|_{2,n} \\
 &\leqslant & \|S_T\|\|\widetilde \delta\| + \|S_{T^c}\|_{\infty} \sqrt{\widehat m}\|\widetilde \delta \| + 2c_s\|\widetilde \delta\|_{2,n} \\
& \leqslant & \displaystyle  \|S_T\|\frac{\|\widetilde \delta\|_{2,n}}{\kappa(\widehat m)} + \|S_{T^c}\|_{\infty} \sqrt{\widehat m}\frac{\|\widetilde \delta\|_{2,n}}{\kappa(\widehat m)} + 2c_s\|\widetilde \delta\|_{2,n}. \\
\end{array}$$

Next, note that for any $j \in \{1,\ldots,p\}$ we have $\Ep[S_j^2] = 4\sigma^2/n$, so that $\Ep[\|S_T\|^2]\leqslant 4\sigma^2s/n$. Thus, by Chebyshev inequality, for any $\tilde \gamma>0$, there is a constant $A_{\tilde \gamma}$ such that $\|S_T\| \leqslant A_{\tilde \gamma} \sigma\sqrt{s/n}$ with probability at least  $1-\tilde \gamma$. Moreover, using Lemma \ref{Lemma:GaussianTail}, $\|S_{T^c}\|_\infty \leqslant A'_{\tilde \gamma} 2\sigma\sqrt{2 \log p \ /n}$ with probability at least $1-\tilde \gamma$ for some constant $A'_{\tilde \gamma}$. Define $A_{\gamma,n} := K_\gamma \sigma\sqrt{(s+\widehat m \log p )/n}$ so that $A_{\gamma,n} \geqslant \|S_T\| + \sqrt{\widehat m} \|S_{T^c}\|_\infty$ with probability at least $1-\gamma$ for some constant $K_\gamma<\infty$ independent of $n$ and $p$.

Combining these relations, with probability at least $1-\gamma$ we have
$$
\|\widetilde \delta\|_{2,n}^2 - A_{\gamma,n} \|\widetilde \delta\|_{2,n}/\kappa(\widehat m) - 2c_s\|\widetilde \delta\|_{2,n}\leqslant  B_n \wedge C_n,
$$
solving which we obtain:
\begin{equation}\label{postLASSOmain}\|\widetilde \delta\|_{2,n} \leqslant A_{\gamma,n}/\kappa(\widehat m)+ 2c_s+ \sqrt{(B_n)_+ \wedge (C_n)_+}.\end{equation}

Note that by the optimality of $\hat \beta$ in the LASSO problem, and letting $\widehat \delta = \widehat \beta - \beta_0$,
\begin{equation}\label{endarray} \begin{array}{rcl} &  \displaystyle B_n = \widehat Q(\hat \beta ) - \widehat Q(\beta_0 ) & \leqslant  \frac{\lambda}{n}(
\| \beta_0\|_{1} - \|\hat\beta\|_{1}) \leqslant \frac{\lambda}{n}( \| \widehat \delta_T \|_{1} - \|\widehat \delta_{T^c}\|_{1}).
\end{array}
\end{equation}
If $\|\widehat\delta_{T^c}\|_{1} > \cc \|\widehat \delta_{T}\|_{1}$, we have $ \displaystyle\hat Q(\hat \beta ) - \hat Q(\beta_0 ) \leqslant 0$ since
$\cc \geqslant 1$. Otherwise, if $\|\widehat\delta_{T^c}\|_{1} \leqslant \cc \|\widehat \delta_{T}\|_{1}$, by RE($c$) we have
\begin{equation}\label{endarray2} \begin{array}{rcl}
& &\displaystyle\displaystyle B_n := \widehat Q(\hat \beta ) - \widehat Q(\beta_0 )  \leqslant \frac{\lambda}{n} \|\widehat \delta_{T}\|_{1} \leqslant \displaystyle \frac{\lambda}{n} \frac{\sqrt{s}\|\widehat \delta \|_{2,n}}{\kappa_\cc}. \\
\end{array}\end{equation}

The choice of $\lambda$ yields $\lambda \geqslant cn\|S\|_\infty$ with probability $1-\alpha$. Thus, by applying Theorem \ref{Thm:Nonparametric}, which requires $\lambda \geqslant cn\|S\|_\infty$, we can bound $\|\widehat\delta\|_{2,n}$.

Finally, with probability $1-\alpha-\gamma$ we have that (\ref{postLASSOmain}) and (\ref{endarray2}) with $\|\widehat\delta\|_{2,n}\leqslant (1+1/c)\lambda\sqrt{s}/n\kappa_\cc + 2c_s$ hold, and the result follows since if $T \subseteq \widehat T$ we have $C_n = 0$ so that $B_n\wedge C_n \leqslant 1\{T\not\subseteq \widehat T\} B_n$.
\end{proof}

\begin{proof}[Theorem \ref{theorem: sigma}]
Consider the case of Post-LASSO; the proof for LASSO is similar.   Consider the case
with $k=1$, i.e. when $\hat \sigma = \hat \sigma^k$ for $k=1$. Then we have
$$
\begin{array}{rl}
\displaystyle \left| \frac{\widehat Q(\widetilde \beta)}{\sigma^2} - \frac{\En[\epsilon_i^2]}{\sigma^2} \right| &  \displaystyle\leqslant \frac{\|\widetilde \beta - \beta_0\|^2_{2,n}}{\sigma^2}
+ \frac{\|S\|_{\infty} \| \widetilde \beta - \beta_0\|_1}{\sigma^2} + \\
& \displaystyle + \frac{2c_s\|\widetilde \beta - \beta_0\|_{2,n}}{\sigma^2} + \frac{2c_s\sqrt{\En[\epsilon_i^2]}}{\sigma^2} +\frac{c_s^2}{\sigma^2} = o_P(1).
\end{array}$$
since  $\|\widetilde \beta - \beta_0\|_{2,n} \lesssim_P \sigma \sqrt{(s/n) \log p}$ by Corollary \ref{corollary3:postrate} and by assumption on $\hat\sigma^0$,
$\|S\|_{\infty} \lesssim_P \sigma \sqrt{(1/n) \log p}$ by Lemma \ref{Lemma:GaussianTail},
 $\| \widetilde \beta - \beta_0\|_1 \leqslant \sqrt{\hat s} \ \|\widetilde \beta - \beta\|_2 \lesssim_P \sqrt{\hat s} \ \|\widetilde \beta - \beta\|_{2, n}$
by condition RSE, $\hat s\lesssim_P s$ by Corollary \ref{corollary2:sparsity} and $c_s \lesssim \sigma \sqrt{s/n}$  by condition ASM,
and $s \log p/n \to 0$ by assumption, and $\frac{\En[\epsilon_i^2]}{\sigma^2} - 1 \to_P 0$ by the Chebyshev inequality.  
Finally,
$ n/(n-\hat s) = 1 + o_P(1)$ since $\hat s\lesssim_P s$ by Corollary \ref{corollary2:sparsity} and $s \log p/n \to 0$.  The result
for $2 \leqslant k \leqslant I-1$ follows by induction.
\end{proof}

\section{Auxiliary Lemmas}
Recall that  $\|S/(2\sigma)\|_\infty = \max_{1 \leqslant j \leqslant p} |\En[x_{ij} g_i]|$, where $g_i$ are i.i.d. $N(0,1)$, for $i =1,...,n,$
conditional on $X =[x_1',...,x_n']'$, and $\En[x_{ij}^2] =1$ for each $j=1,...,p$, and
note that $ P( n \|S/(2\sigma) \|_{\infty} \geqslant \Lambda(1-\alpha |X)|X) = \alpha$ by definition.

\begin{lemma}\label{Lemma:GaussianTail} We have that for $t \geqslant 0$:
\begin{eqnarray*}
&& P( n \|S/(2\sigma) \|_{\infty} \geqslant t \sqrt{n} |X) \leqslant 2p (1- \Phi(t)) \leqslant  2p \frac{1}{t} \phi(t),\\
&& \Lambda(1-\alpha|X) \leqslant \sqrt{n} \Phi^{-1}(1-\alpha/2p) \leqslant \sqrt{2 n \log(2p/\alpha)}, \\
&&  P( n \|S/(2\sigma) \|_{\infty} \geqslant  \sqrt{2 n \log(2p/\alpha)}  |X) \leq \alpha.
\end{eqnarray*}
\end{lemma}
\begin{proof}
To establish the first claim, note that
 $\sqrt{n} \|S/2\sigma \|_{\infty} = \max_{1\leqslant j \leqslant p} |Z_j|$, where $Z_j= \sqrt{n} \En[x_{ij} g_i]$ are $N(0,1)$
 by $g_i$ i.i.d. $N(0,1)$  conditional on $X$ and by $\En[x_{ij}^2] =1$ for each $j=1,...,p$.
 Then the first claim follows by observing that for $z \geqslant 0$ by the union bound
 $ P ( \max_{1\leqslant j \leqslant p} |Z_j| > z) \leqslant p   P (|Z_j| > z) = 2p (1- \Phi(z))$
 and by
$
  (1- \Phi(z)) = \int_{z}^{\infty} \phi(u) du \leqslant  \int_{z}^{\infty} (u/z) \phi(u) dz \leqslant (1/z) \phi(z).
$ The second and third claim follow by noting that  $2p (1- \Phi(t')) = \alpha$ at $t' = \Phi^{-1}(1-\alpha/2p)$, and $2p \frac{1}{t''} \phi(t'') = \alpha$ at $t'' \leqslant \sqrt{2\log (2p/\alpha)}$, so that, in view of the first claim, $\Lambda(1-\alpha |X) \leqslant \sqrt{n} t' \leqslant \sqrt{n} t''$.
\end{proof}

~\\

\begin{lemma}[Sub-linearity of restricted sparse eigenvalues]
\label{Lemma:SparseEigenvalueIMP}For any integer $k \geqslant 0$ and constant $\ell \geqslant 1$ we have
$ \phi(\ceil{\ell k}) \leqslant  \lceil \ell \rceil \phi(k).$
\end{lemma}
\begin{proof}
Let $W := \En[x_ix_i']$ and $\bar \alpha$ be such that $\phi(\ceil{\ell k})= \bar \alpha' W \bar\alpha $, $\|\bar\alpha\|=1$. We can decompose the vector $\bar \alpha$ so that
$$\bar \alpha = \sum_{i=1}^{\ceil{\ell}} \alpha_i, \text{ with }
\sum_{i=1}^{\ceil{\ell}} \|\alpha_{iT^c}\|_0  = \|\bar
\alpha_{T^c}\|_0 \text{ and } \alpha_{iT} = \bar \alpha_T / \ceil{\ell},$$ where we can choose $\alpha_i$'s such that $\|\alpha_{iT^c}\|_0
\leqslant k$ for each $i=1,...,\ceil{\ell}$, since $\lceil \ell \rceil k \geqslant \ceil{\ell k}$. Note that the vectors $\alpha_i$'s have no overlapping support outside $T$. Since $W$ is
positive semi-definite, $\alpha_i'W\alpha_i +
\alpha_j'W\alpha_j \geqslant 2\left|\alpha_i'W\alpha_j\right|$  for any pair $(i, j)$.
Therefore
$$
\begin{array}{rcl}
& &  \phi(\ceil{\ell k}) =   \bar \alpha' W \bar\alpha \ = \  \displaystyle \sum_{i=1}^{\lceil \ell \rceil} \sum_{j= 1}^{\lceil \ell \rceil} \alpha_i'W\alpha_j  \\ &  & \leqslant  \displaystyle \ \sum_{i=1}^{\lceil \ell \rceil} \sum_{j= 1}^{\ceil{\ell}} \frac{\alpha_i'W\alpha_i +  \alpha_j'W\alpha_j}{2} =   \displaystyle\lceil \ell \rceil \sum_{i=1}^{\lceil \ell \rceil} \alpha_i'W\alpha_i \\
   &  & \leqslant   \displaystyle {\lceil \ell \rceil} \sum_{i=1}^{\lceil
\ell \rceil} \|\alpha_i\|^2 \phi(\|\alpha_{iT^c}\|_0)
 \leqslant  \displaystyle {\lceil \ell \rceil} \max_{i=1,\ldots,{\lceil
\ell \rceil}} \phi(\|\alpha_{iT^c}\|_0) \leqslant {\lceil \ell \rceil} \phi(k),
\end{array}
$$ where we used that $$\sum_{i=1}^{\lceil \ell \rceil} \|\alpha_i\|^2=\sum_{i=1}^{\lceil \ell \rceil} (\|\alpha_{iT}\|^2+\|\alpha_{iT^c}\|^2) = \frac{\|\bar\alpha_{T}\|^2}{\ceil{\ell}} + \sum_{i=1}^{\lceil \ell \rceil} \|\alpha_{iT^c}\|^2  \leqslant \|\bar\alpha\|^2 = 1.$$ \end{proof}

\begin{lemma}\label{Lemma:BoundKAPPA}
Let $\cc = (c+1)/(c-1)$ we have for any integer $m > 0$
$$ \kappa_\cc \geqslant \kappa(m) \left( 1 - \mmu{m}\cc \sqrt{\frac{s}{m}}\right).$$
\end{lemma}
\begin{proof}
We follow the proof in \cite{BickelRitovTsybakov2009}. Pick an arbitrary vector $\delta$ such that $\|\delta_{T^c}\|_1\leqslant \cc \|\delta_T\|_1$. Let $T^1$ denote the $m$ largest components of $\delta_{T^c}$. Moreover, let $T^c = \cup_{k=1}^{K} T^k$ where $ K = \ceil{(p-s)/m}$, $|T^k|\leqslant m$ and $T^k$ corresponds to the $m$ largest components of $\delta$ outside $T\cup (\cup_{d=1}^{k-1} T^d)$.

We have
$$ \begin{array}{rl}
 \|\delta\|_{2,n} \geqslant \|\delta_{T\cup T^1}\|_{2,n} -  \|\delta_{(T\cup T^1)^c}\|_{2,n} & \displaystyle  \geqslant \kappa(m)\|\delta_{T\cup T^1}\| -  \sum_{k=2}^{K} \|\delta_{T^k}\|_{2,n} \\
 & \displaystyle  \geqslant \kappa(m)\|\delta_{T\cup T^1}\| - \sqrt{\phi(m)} \sum_{k=2}^{K} \|\delta_{T^k}\|.\end{array}$$

Next note that
$$ \|\delta_{T^{k+1}}\| \leqslant \|\delta_{T^k}\|_1/\sqrt{m}.$$
Indeed, consider the problem $\max \{\|v\|/\|u\|_1 : v, u \in \RR^m,  \max_i |v_i| \leqslant \min_i |u_i|\}$. Given a $v$ and $u$ we can always increase the objective function by using $\tilde v = \max_i|v_i| (1,\ldots,1)'$ and $\tilde u'=\min_i|u_i|(1,\ldots,1)'$ instead. Thus, the maximum is achieved at $v^*=u^*=(1,\ldots,1)'$, yielding $1/\sqrt{m}$.

Thus, by $\|\delta_{T^c}\|_1\leqslant \cc \|\delta_T\|_1$ and $|T|= s$
$$ \sum_{k=2}^K\| \delta_{T^k} \| \leqslant  \sum_{k=1}^{K-1}  \frac{\|\delta_{T^k}\|_1}{\sqrt{m}} \leqslant \frac{\|\delta_{T^c}\|_1}{\sqrt{m}} \leqslant \cc \|\delta_{T}\|\sqrt{\frac{s}{m}} \leqslant \cc \|\delta_{T\cup T^1 }\|\sqrt{\frac{s}{m}}. $$
Therefore, combining these relations with $\|\delta_{T\cup T^1}\|\geqslant \|\delta_T\| \geqslant \|\delta_T\|_1/\sqrt{s}$ we have
$$  \|\delta\|_{2,n}  \geqslant \frac{\|\delta_T\|_1}{\sqrt{s}} \kappa(m)\left( 1 -
\mmu{m} \cc \sqrt{s/m} \right) $$
which leads to
$$ \frac{\sqrt{s}\|\delta\|_{2,n}}{\|\delta_T\|_1} \geqslant \kappa(m)\left( 1 - \mmu{m}\cc \sqrt{s/m} \right).$$

\end{proof}

\bibliographystyle{spmpsci}

\bibliography{biblioLectureNotes}

\begin{thebibliography}{10}
\providecommand{\url}[1]{{#1}}
\providecommand{\urlprefix}{URL }
\expandafter\ifx\csname urlstyle\endcsname\relax
  \providecommand{\doi}[1]{DOI~\discretionary{}{}{}#1}\else
  \providecommand{\doi}{DOI~\discretionary{}{}{}\begingroup
  \urlstyle{rm}\Url}\fi

\bibitem{Akaike1974}
Akaike, H.: A new look at the statistical model identification.
\newblock MIEEE Transactions on Automatic Control \textbf{AC-19}, 716--–723
  (1974)

\bibitem{ACF2006}
Angrist, J., Chernozhukov, V., Fernández-Val, I.: Quantile regression under
  misspecification, with an application to the {U}.{S}. wage structure.
\newblock Econometrica \textbf{74}(2), 539--563 (2006)

\bibitem{AK1991}
Angrist, J.D., Krueger, A.B.: Does compulsory school attendance affect
  schooling and earnings?
\newblock The Quarterly Journal of Economics \textbf{106}(4), 979--1014 (1991)

\bibitem{BarroLee1994}
Barro, R.J., Lee, J.W.: Data set for a panel of 139 countries.
\newblock NBER, http://www.nber.org/pub/barro.lee.html  (1994)

\bibitem{BarroSala1995}
Barro, R.J., {Sala--i--Martin}, X.: Economic Growth.
\newblock McGraw-Hill, New York (1995)

\bibitem{BCCH-LASSOIV}
Belloni, A., Chen, D., Chernozhukov, V., Hansen, C.: Sparse models and methods
  for optimal instruments with an application to eminent domain.
\newblock arXiv:[math.ST]  (2010)

\bibitem{BC-PostLASSO}
Belloni, A., Chernozhukov, V.: Post-$\ell_1$-penalized estimators in
  high-dimensional linear regression models.
\newblock arXiv:[math.ST]  (2009)

\bibitem{BickelRitovTsybakov2009}
Bickel, P.J., Ritov, Y., Tsybakov, A.B.: Simultaneous analysis of {LASSO} and
  {D}antzig selector.
\newblock Annals of Statistics \textbf{37}(4), 1705--1732 (2009)

\bibitem{CandesPlan2009}
Cand\`{e}s, E.J., Plan, Y.: Near-ideal model selection by $l_1$ minimization.
\newblock Ann. Statist. \textbf{37}(5A), 2145--2177 (2009)

\bibitem{CandesTao2007}
Cand\`{e}s, E.J., Tao, T.: The {D}antzig selector: statistical estimation when
  p is much larger than n.
\newblock Ann. Statist. \textbf{35}(6), 2313--2351 (2007)

\bibitem{DonohoJohnstone1994}
Donoho, D.L., Johnstone, J.M.: Ideal spatial adaptation by wavelet shrinkage.
\newblock Biometrika \textbf{81}(3), 425--455 (1994)

\bibitem{GeJiangYe2010}
Ge, D., Jiang, X., Ye, Y.: A note on complexity of $l_p$ minimization.
\newblock Stanford Working Paper  (2010)

\bibitem{vdGeer}
van~de Geer, S.A.: High-dimensional generalized linear models and the lasso.
\newblock Annals of Statistics \textbf{36}(2), 614--645 (2008)

\bibitem{hhn:weakiv}
Hansen, C., Hausman, J., Newey, W.K.: Estimation with many instrumental
  variables.
\newblock Journal of Business and Economic Statistics \textbf{26}, 398--422
  (2008)

\bibitem{Koltchinskii2009}
Koltchinskii, V.: Sparsity in penalized empirical risk minimization.
\newblock Ann. Inst. H. Poincaré Probab. Statist. \textbf{45}(1), 7--57 (2009)

\bibitem{OneMillion}
Levine, R., Renelt, D.: A sensitivity analysis of cross-country growth
  regressions.
\newblock The American Economic Review \textbf{82}(4), 942--963 (1992)

\bibitem{Lounici2008}
Lounici, K.: Sup-norm convergence rate and sign concentration property of lasso
  and dantzig estimators.
\newblock Electron. J. Statist. \textbf{2}, 90--102 (2008)

\bibitem{MY2007}
Meinshausen, N., Yu, B.: Lasso-type recovery of sparse representations for
  high-dimensional data.
\newblock Annals of Statistics \textbf{37}(1), 2246--2270 (2009)

\bibitem{Natarajan1995}
Natarajan, B.K.: Sparse approximate solutions to linear systems.
\newblock SIAM Journal on Computing \textbf{24}, 227--234 (1995)

\bibitem{RigolletTsybakov2010}
Rigollet, P., Tsybakov, A.B.: Exponential screening and optimal rates of sparse
  estimation.
\newblock ArXiv:1003.2654  (2010)

\bibitem{RudelsonVershynin2008}
Rudelson, M., Vershynin, R.: On sparse reconstruction from fourier and gaussian
  measurements.
\newblock Communications on Pure and Applied Mathematics \textbf{61}, 1025–1045
  (2008)

\bibitem{TwoMillion}
{Sala--i--Martin}, X.: I just ran two million regressions.
\newblock The American Economic Review \textbf{87}(2), 178--183 (1997)

\bibitem{Schwarz1978}
Schwarz, G.: Estimating the dimension of a model.
\newblock Annals of Statistics pp. 461–--464 (1978)

\bibitem{T1996}
Tibshirani, R.: Regression shrinkage and selection via the lasso.
\newblock J. Roy. Statist. Soc. Ser. B \textbf{58}, 267--288 (1996)

\bibitem{Wainright2006}
Wainwright, M.: Sharp thresholds for noisy and high-dimensional recovery of
  sparsity using $\ell_1$-constrained quadratic programming (lasso).
\newblock IEEE Transactions on Information Theory \textbf{55}, 2183--2202
  (2009)

\bibitem{ZhangHuang2006}
Zhang, C.H., Huang, J.: The sparsity and bias of the lasso selection in
  high-dimensional linear regression.
\newblock Ann. Statist. \textbf{36}(4), 1567--1594 (2008)

\bibitem{ZhaoYu2006}
Zhao, P., Yu, B.: On model selection consistency of lasso.
\newblock J. Machine Learning Research \textbf{7}, 2541--2567 (2006)

\end{thebibliography}

\end{document}